\documentclass[opre,nonblindrev]{informs3}

\usepackage[utf8]{inputenc}

\usepackage{amsfonts,color}
\usepackage[usenames,dvipsnames,svgnames,table]{xcolor}
\definecolor{darkgreen}{rgb}{0.0,0,0.9}
\usepackage{authblk}
\usepackage{comment}
\usepackage{dsfont}
\usepackage[basic]{complexity}
\usepackage[linesnumbered,ruled,vlined]{algorithm2e}
\usepackage[pdfstartview=FitH, hidelinks]{hyperref}
\usepackage[capitalize,nameinlink]{cleveref}
\usepackage{tcolorbox}
\usepackage[short,c2]{optidef}

\usepackage{multirow}
\usepackage{paralist}
\usepackage{booktabs}
\usepackage{graphicx}
\usepackage{subfig}
\usepackage[export]{adjustbox}

\OneAndAHalfSpacedXI 

\usepackage{endnotes}
\let\footnote=\endnote
\let\enotesize=\normalsize
\def\notesname{Endnotes}%
\def\makeenmark{\hbox to1.275em{\theenmark.\enskip\hss}}
\def\enoteformat{\rightskip0pt\leftskip0pt\parindent=1.275em
  \leavevmode\llap{\makeenmark}}

\usepackage{natbib}
 \bibpunct[, ]{(}{)}{,}{a}{}{,}%
 \def\bibfont{\small}%
\TheoremsNumberedThrough     
\ECRepeatTheorems

\EquationsNumberedThrough    

\newcommand*{\bbar}[1]{\bar{\bar{#1}}}

\let\R\relax
\newcommand*{\R}{\mathbb{R}}

\newcommand*{\Rplus}{\mathbb{R_+}}
\newcommand*{\Qplus}{\mathbb{Q_+}}

\newcommand*{\suppress}[1]{}

\def\bvec#1{{\boldsymbol #1}}
\newcommand*{\vecz}{\bvec{0}}

\def\range#1{[#1]}

\def\calO{\mathcal{O}}

\def\calX{\mathcal{X}}
\newcommand{\CN}{\mbox{${\cal N}$}}
\newcommand{\CS}{\mbox{${\cal S}$}}

\DeclareMathOperator{\prx}{prox}

\newcommand\bbeta{\boldsymbol{\beta}}

\newcommand\llambda{\boldsymbol{\lambda}}

\newcommand{\prf}[1]{\proof{Proof.}#1\Halmos
\endproof}

\newcommand\cb{\boldsymbol{\mathit{c}}}

\renewcommand\ee{\boldsymbol{\mathit{e}}}

\newcommand\gb{\boldsymbol{\mathit{g}}}

\newcommand\pp{\boldsymbol{\mathit{p}}}

\newcommand\rr{\boldsymbol{\mathit{r}}}

\newcommand\vv{\boldsymbol{\mathit{v}}}

\newcommand\yy{\boldsymbol{\mathit{y}}}
\newcommand\zz{\boldsymbol{\mathit{z}}}
\newcommand\xx{\boldsymbol{\mathit{x}}}

\begin{document}

\RUNAUTHOR{Hosseini and Vazirani}
\RUNTITLE{Nash-Bargaining-Based Models for Matching Markets}
\TITLE{Nash-Bargaining-Based Models for \\
Matching Markets: One-Sided and Two-Sided; \\
Fisher and Arrow-Debreu}

\ARTICLEAUTHORS{%
\AUTHOR{Mojtaba Hosseini}
\AFF{Tippie College of Business, University of Iowa, IA \EMAIL{mojtaba-hosseini@uiowa.edu}} 
\AUTHOR{Vijay V. Vazirani}
\AFF{Computer Science Department, University of California, Irvine, CA, \EMAIL{vazirani@ics.uci.edu}}
} 

\ABSTRACT{%
This paper addresses two deficiencies of models in the area of matching-based market design. The first arises from the recent realization that the most prominent solution that uses cardinal utilities, namely the Hylland-Zeckhauser (HZ) mechanism \citep{hylland}, is intractable;  computation of even an approximate equilibrium is PPAD-complete \citep{VY-HZ, HZ-hardness}. The second is the extreme paucity of models that use cardinal utilities.
Our paper addresses both these issues by proposing Nash-bargaining-based matching market models. Since the Nash bargaining solution is captured by a convex program, efficiency follows. In addition, it possesses several desirable game-theoretic properties. Our approach yields a rich collection of models: for one-sided as well as two-sided markets, for Fisher as well as Arrow-Debreu settings, and for a wide range of utility functions, all the way from linear to Leontief. 
We give very fast implementations for these models using Frank-Wolfe and Cutting Plane algorithms. These help solve large instances with several thousand agents and goods in a matter of minutes on a PC, even for a one-sided matching market under piecewise-linear concave utility functions and a two-sided matching market under linear utility functions. In contrast, using HZ, going beyond even $n = 10$ is prohibitive. Several new ideas were needed, beyond the standard methods, to obtain these implementations. In particular, we present several lower bounding schemes, which not only help improve the convergence of our solution methods but also shed light on fairness properties of the Nash-bargaining-based models.

}%


\KEYWORDS{Matching-based market design, Nash bargaining, convex optimization, Frank-Wolfe algorithm, cutting planes, general equilibrium theory, one-sided markets, two-sided markets} \HISTORY{}

\maketitle

  \section{Introduction}
\label{sec:intro}

The recent computer science revolutions of the Internet and mobile computing led to the launching of highly impactful and innovative one-sided and two-sided matching markets, such as online advertisement platforms (Google Ads), vacation rentals (Airbnb, VRBO), ride-hailing (Uber, Lyft), food delivery (Uber Eats, Postmates, GrubHub, Doordash, Instacart), freelancing and employment services (Taskrabbit, Upwork, Fiverr, LinkedIn), on-demand beauty services (Glamsquad, PRIV), and online dating services (Match.com, OkCupid).
In turn, they led to a major revival of the area of matching-based market design, much of which is expounded in the upcoming book \cite{Echenique2023online}; see also \citep{Simons, va.talk, bimpikis2019spatial, shi2022optimal}. It is safe to assume that innovations will keep coming in the future and that new models and efficient mechanisms, with good properties, will be needed in the future. 

Within the area of matching-based market design, the most prominent solution that uses cardinal utilities is the Hylland-Zeckhauser (HZ) mechanism \citep{hylland} for a one-sided matching markets under linear utilities. (For a brief comparison of cardinal and ordinal utilities for matching markets, see Section \ref{sec.related}.)  HZ is based on  creating parity between demand and supply, i.e., it uses the power of a pricing mechanism, which gives it attractive properties: the allocations produced satisfy Pareto optimality and envy-freeness \citep{hylland} and the mechanism is incentive compatible in the large \citep{He2018pseudo}. 

A serious drawback of HZ, from the viewpoint of practical applicability, is lack of computational efficiency: the recent papers \cite{VY-HZ} and \cite{HZ-hardness} show that the problem of computing even an approximate equilibrium is PPAD-complete. More precisely, \cite{VY-HZ} showed membership in PPAD and remarked that it will not be surprising if intractability sets in even for the highly special case in which utilities of agents come from a tri-valued set, say $\{0, \frac{1}{2}, 1\}$; for bi-valued sets, they gave an efficient algorithm. Next, \cite{HZ-hardness} showed PPAD-hardness even for the case that utilities of agents come from a four-valued set; the tri-valued case is open. 

Hylland and Zeckhauser \citep{hylland} also studied the Arrow-Debreu extension of their model, in which agents have initial endowments of goods; however, they ended their investigation on finding instances that do not admit an equilibrium. In light of this difficulty, studying further generalizations made little sense. In particular, we are not aware of any two-sided matching market models that use cardinal utilities. This stands in sharp contrast with general equilibrium theory, which has defined and extensively studied several fundamental market models to address a number of specialized and realistic situations. That leads to the second issue addressed by our paper, namely the extreme paucity of models that use cardinal utilities. 

Our paper addresses both these issues by proposing Nash-bargaining-based matching market models. As is well known, the Nash bargaining solution is captured as an optimal solution to a convex program. If for such a program, a separation oracle can be implemented in polynomial time, then using the ellipsoid algorithm, one can get as good an approximation as desired in time that is polynomial in the number of bits of accuracy required \citep{GLS, Vishnoi.book}. For all models defined in this paper, the constraints of the convex programs are linear, thereby ensuring zero duality gap and easy solvability. 

A second gain from the move to Nash bargaining is that it yields a plethora of matching market models, not only one-sided but also two-sided. Moreover, all models work for Fisher as well as  Arrow-Debreu settings; interestingly enough, the latter are not much harder than the former. We present five one-sided market models, covering a large range of utility functions, all the way from linear to Leontief. When generalized appropriately, all these models admit counterparts in the two-sided setting as well; to illustrate the style of generalization, we present a two-sided market model with linear utilities only.

The game-theoretic properties of the Nash bargaining solution include: it satisfies Pareto optimality and symmetry, and since it maximizes the product of utilities of agents, the allocations it produces are remarkably fair. The latter has been noted by several researchers \citep{bertsimas2011price, Nash-Unreasonable, Abebe-MM-Truthful, Moulin2018fair} and has been further explored under the name of \textit{Nash Social Welfare} \citep{cole2017convex,cole2018approximating,branzei2022nash}. An important property that has been sacrificed, compared to HZ, is envy-freeness. In complete absence of envy-freeness, Pareto optimality has little meaning, since then one agent can be given the most desirable goods. However, as shown in \cite{Nash-Combinatorial}, the Nash-bargaining-based models do satisfy an alternative property called the {\em equal-share fairness property}, which disallows highly skewed allocations. Under this property, for linear utility functions, each agent $i$ must get at least $\frac{1}{2n}$ fraction of the {\em total} utility which $i$ could have gotten by being allocated {\em all} of the goods. In Section \ref{sec.CPs-bounding} we show the same result for more general concave utility functions, and show that strictly tighter bounds exist for several matching markets.

The following question arises: Is the shift from pricing to Nash bargaining a sound one, i.e., is there a fundamental connection between the two types of models? Based on \cite{Nash-Combinatorial}, Section \ref{sec.connection} provides such a connection via the celebrated Eisenberg-Gale convex program.

We note that the origins of the idea of operating markets via Nash bargaining go back to \cite{va.rational}. For the linear case of the Arrow-Debreu market model, instead of seeking allocations via a pricing mechanism, \cite{va.rational} formulated it as a Nash bargaining game and gave a combinatorial, polynomial time algorithm for solving the underlying convex program. In contrast, under HZ, equilibrium prices are not captured by any known mathematical construct, regardless of its computational complexity. Furthermore, the only known method for conducting an exhaustive search for obtaining an HZ equilibrium is algebraic cell decomposition \citep{Basu1995}; its use for computing HZ equilibria was studied in \cite{HZ-Algebraic-Cell}. Each iteration of this method is time-consuming, and as a result, HZ is viable for only very small values of $n$, not exceeding 10. 

As is well known, polynomial time solvability is often just the beginning of the process of obtaining an ``industrial grade''  implementation. Towards this end, we give very fast implementations as well as experimental results for all five of our one-sided market models and the most basic two-sided model. In particular, our implementation can solve very large instances, with several thousand agents and goods, in a matter of minutes even for a one-sided matching market under piecewise-linear concave utility functions and a two-sided matching market under linear utility functions. In Section \ref{sec.ideas} we have described how the standard methods needed to be adapted to the special intricacies of our convex programming formulations, in order to obtain these very fast implementations.  



\subsection{A Connection Between HZ and Nash-Bargaining-Based Models}
\label{sec.connection}

In this section, we attempt a comparative study of one-sided matching markets under the two types of mechanisms, pricing and Nash bargaining. The following two questions arise:
\begin{enumerate}
	\item  Is this shift from a pricing to a Nash bargaining mechanism a sound one, i.e., is there a fundamental connection between the resulting types of models? 
	\item  Is either type of model reducible to the other?
\end{enumerate}

The answer to the second question is ``No'' since under an affine transformation of the utility function of an agent, the Nash bargaining solution and an HZ equilibrium change in different ways: The former solution undergoes the same affine transformation, as stated by one of the axioms of Nash bargaining, see Section \ref{sec.Nash}. On the other hand, the latter remains unchanged; this was shown in \cite{VY-HZ} via the following statement. Let agent $i$'s utility function be $u_i = \{u_{i1}, u_{i12}, \ldots, u_{in} \}$. For two numbers $s > 0$ and $h \geq 0$, define $u_i' = \{u_{i1}', u_{i12}', \ldots, u_{in}' \}$ as follows: $\forall j \in G, \ u_{ij}' = s \cdot u_{ij} + h$. Let $I'$ be the instance obtained by replacing $u_i$ by $u_i'$ in $I$, keeping utilities of other agents unchanged. Then $(\pp, \xx)$ is an equilibrium for $I$ if and only if it is for $I'$.

The answer to the first question is ``Yes'', due to the connection established in \cite{Nash-Combinatorial}. We provide a brief synopsis of the argument below. 

First consider the linear Fisher market model defined in Section \ref{sec.Fisher}. The setup of the {\em linear Fisher problem (LFP)} is identical, except that the agents don't have any money, so this is not really a market model. The problem is to design a polynomial time mechanism for distributing all the goods among the agents so that the allocation satisfies Pareto optimality.

\begin{maxi}
	{} {\sum_{i \in A}  {\log \left(\sum_{j \in G}  {u_{ij} x_{ij}}\right)}}
	{\label{eq.EG-LF}}
	{}
	\addConstraint{\sum_{i \in A} {x_{ij}}}{\leq 1 }{\quad \forall j \in G}
	\addConstraint{\xx}{\geq 0.}
\end{maxi}

\cite{Nash-Combinatorial} give two such mechanisms. The first is to give each agent 1 Dollar, thereby transforming LFP to the linear Fisher market model, and ask for an equilibrium allocation  satisfying Pareto optimality. This can be obtained in polynomial time, via a combinatorial algorithm \citep{DPSV}, or as an optimal solution to the celebrated Eisenberg-Gale convex program \citep{eisenberg}, given in (\ref{eq.EG-LF}).

The second is to view LFP as a Nash bargaining problem; Pareto optimality is one of the axioms which it satisfies, see Section \ref{sec.Nash}. This is done by defining a convex, compact set $\CN \subseteq \R_+^n$, called the feasible set, and a point $\cb \in \CN$, called the disagreement point, see Section \ref{sec.Nash} for details. In this case, $\cb = \vecz$, and $\CN$ will consist of all possible vectors of utilities to the $n$ agents that can be obtained by partitioning 1 unit each of all $m$ goods among the agents. It is easy to see that the resulting convex program will be precisely Eisenberg-Gale convex program. Therefore, the two mechanisms are identical!

Next, \cite{Nash-Combinatorial} define the {\em linear Fisher unit demand problem (LFUP)} to be LFP with the additional requirements that $m = n$ and that each agent should get a total of one unit of goods. As a result, every feasible allocation is a fractional perfect matching over the $n$ agents and $n$ goods.
Now it turns out that when LFUP is solved via the pricing mechanism, it is identical to HZ, and when it is solved via the Nash bargaining mechanism it is identical to {\em 1LF}, i.e., our most basic Nash-bargaining-based model, see Section \ref{sec.1-models}. This establishes a connection between HZ and the Nash-bargaining-based models; it is illustrated in Figure \ref{fig:main}.


\begin{figure}
	\centering
	\includegraphics[width=0.65\linewidth]{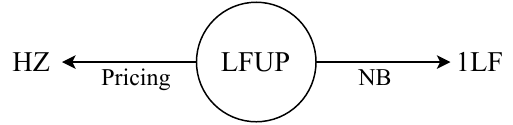}
	\caption{Figure illustrating connection between HZ and NB.}
	\label{fig:main}
\end{figure}

\subsection{Ideas Needed Beyond Standard Methods}
\label{sec.ideas}

Our solution methods, namely Frank-Wolfe (FW) algorithm and cutting-plane (CP) algorithm, rely on iterative linear approximations of convex programs for the one-sided and two-sided market models. For efficient implementation of these algorithms, one needs to pay attention to the structural properties of these models as described below.

We implement a FW algorithm for solving instances of the matching markets with general utility functions. Given that the feasible region of our convex programs corresponds to a matching polytope, we exploit this property to produce the ``atom'' solutions in FW efficiently. The solution produced by FW is therefore a sparse convex combination of a set of integral perfect matchings. For piecewise linear utilities, we present smooth counterparts and efficient procedures for deriving them. For matching market models over multiple goods, we adapt a Sinkhorn-type matrix scaling algorithm to derive the atom solutions efficiently.


As a benchmark, we also implement a central cutting-plane algorithm, which produces more effective cuts compared to a vanilla cutting-planes, since central points are more likely to be in the relative interior of the feasible region. We additionally implement several enhancement techniques including cut loop management, variable elimination, upper and lower bounding schemes.

A high quality initial solution is also crucial for fast convergence of both algorithms. We show how such solutions can be produced by solving matching problems induced by convexity properties of our matching market models. We also show that when the optimal allocation is integral, the initial solution is optimal, thus both CP and FW terminate after only one iteration.

Finally, we introduce bounding schemes that improve the equal-share bounds and extend them to general utility functions. Besides their game theoretic implications, we use these bounds to alleviate the numerical issues due to the logarithmic form of the objective functions in our convex programs. For FW, we use these bounds to obtain Lipschitz continuous gradients, thus improve the convergence rate of FW.
The logarithmic form of the objective functions also require positive utilities for each agent at each iteration of CP. However, since CP is an outer-approximation algorithm, it is possible that in an iteration of CP, the utilities of some agents may become zero. This makes the solution unbounded, and one cannot extract a cut based on this solution. We resolve this issue by adding optimality lower bounds to the linear programs associated with CP. 

\subsection{Our Results}
\label{sec.results}

In Section \ref{sec.1-models}, we give five basic models for one-sided matching markets covering a wide range of utility functions. For each model, we also give a natural application. In Section \ref{sec.2-models} we give a model for the most basic two-sided matching market, which can be easily enhanced to five more models in a manner analogous to the other five one-sided matching market models given in Section \ref{sec.1-models}. 

In Section \ref{sec.CPs}, we give convex programs capturing the Nash-bargaining-based solution for all the models mentioned above. These convex programs can be solved to $\epsilon$ precision in time that is polynomial in the size of the input and $\log{1/\epsilon}$ via ellipsoid-based methods \citep{GaleS, Vishnoi.book}.  We also provide optimality bounds for these problems, which extend the \textit{equal-share fairness property} to general utility functions and substantially improve them for the linear case. We further use these bounds for enhancing the convergence rate of our solution methods.

In Section \ref{sec.SM}, we present two solution schemes based on Frank-Wolfe method \citep{frank1956algorithm, jaggi2013revisiting} and central cutting-plane method \citep{elzinga1975central}. Both methods rely on linear approximations of the convex programs; the former is solver-free and purely combinatorial, whereas the latter uses an LP solver. We present enhancement techniques as well as an overview of the way structural properties of these problems can be exploited. 

We present computational experiments in Section \ref{sec.Exp} and assess effectiveness of our methods in handling large-scale instances. We demonstrate that our FW implementation can solve instances with up to $n=20,000$ agents/goods in matching markets with linear utilities, even for two-sided, in a matter of minutes. Our tests on matching markets with separable and nonseparable piecewise linear utility functions, as well as markets over multiple goods, further highlight efficiency of our solution methods in handling large instances of these problems. Section~\ref{sec:conclusion} concludes this paper.

An earlier version of this paper was presented in the 13$^{\text{th}}$ Innovations in Theoretical Computer Science Conference \citep{hosseini2022nash}. We extend the theoretical, algorithmic and experimental results in several directions, including but not limited to: (\textit{i}) We introduce new Nash-bargaining-based models over multiple goods and present a tailored generalized Frank-Wolfe algorithm leveraging the Sinkhorn-Knopp's matrix scaling algorithm. (\textit{ii}) We present several optimality bounds for our models, and show that they extend and strictly improve the equal-share matching bound. We also use these bounds to improve convergence of our solution methods. (\textit{iii}) We present Frank-Wolfe algorithms for all of our six matching market models and introduce new enhancement techniques for the Cutting Plane algorithm. (\textit{iv}) We enhance FW by orders of magnitude via a fast implementation of the Auction algorithm. (\textit{v}) We extend our computational experiments and highlight efficiency of our algorithms in handling very large-scale instances.

\subsection{Related Results}
\label{sec.related}

Recently \cite{VY-HZ} undertook the first comprehensive study of the computational complexity of HZ.  They gave an example which has only irrational equilibria; as a consequence, this problem is not in PPAD. They showed that membership of the exact equilibrium computation problem is in the class FIXP and approximate equilibrium is in the class PPAD. They also gave a combinatorial, strongly polynomial time algorithm for computing an equilibrium for the case of dichotomous utilities, 
and they extended this result to the case of bi-valued utilities, i.e., each agent's utility for individual goods comes from a set of cardinality two, though the sets may be different for different agents. Next, \cite{HZ-hardness} showed PPAD-hardness of approximate HZ equilibrium even if utilities of agents come from a four-valued set; the tri-valued case is open. 

The success of our implementations, using available solvers, naturally raises the question of finding efficient combinatorial algorithms with low running times for our proposed market models. \cite{Nash-Combinatorial} has given such algorithms, based on the techniques of multiplicative weights update (MWU) and conditional gradient descent (CGD), for several of our one-sided and two-sided models. They also defined and developed algorithms for the non-bipartite matching market model; this has applications to the roommate problem. Lastly, they gave the connection between HZ and the Nash-bargaining-based models stated in Section \ref{sec.connection}. 

The extension of one-sided matching markets to the setting in which agents have initial endowments of goods, called the Arrow-Debreu setting, has several natural applications beyond the original Fisher setting, e.g., allocating students to rooms in a dorm for the next academic year, assuming their current room is their initial endowment. The issue of obtaining such an extension of the HZ scheme was studied by Hylland and Zeckhauser. However, this culminated in an example which inherently does not admit an equilibrium \citep{hylland}.

As a recourse, \cite{Echenique2019constrained} introduced the notion of an {\em $\alpha$-slack Walrasian equilibrium}. This is a hybrid between the Fisher and Arrow-Debreu settings in which agents have initial endowments of goods and for a fixed $\alpha \in (0, 1]$, the budget of each agent, for given prices of goods, is $\alpha + (1 - \alpha) \cdot m$, where $m$ is the value for her initial endowment. Via a non-trivial proof, using the Kakutani Fixed Point Theorem, they proved that an $\alpha$-slack equilibrium always exists. A pure Arrow-Debreu model was proposed in \cite{Garg-ADHZ} by suitably relaxing the notion of an equilibrium to an {\em $\epsilon$-approximate equilibrium}. Their proof of existence of equilibrium follows from that of \cite{Echenique2019constrained}. 

An interesting recent paper \cite{Abebe-MM-Truthful} defines the notion of a random partial improvement mechanism for a one-sided matching market. This mechanism truthfully elicits the cardinal preferences of the agents and outputs a distribution over matchings that approximates every agent’s utility in the Nash bargaining solution.

In recent years, several researchers have proposed Hylland-Zeckhauser-type mechanisms for a number of applications \citep[e.g., see ][]{Budish2011combinatorial, He2018pseudo, Le2017competitive, Mclennan2018efficient}. The basic scheme has also been generalized in several different directions, including two-sided matching markets, adding quantitative constraints, and to the setting in which agents have initial endowments of goods instead of money, see  \citep{Echenique2019constrained, Echenique2019fairness}.
  
\paragraph{Ordinal vs Cardinal utilities:}
Under ordinal utilities, the agents provide a total preference order over the goods and under cardinal utilities, they provide a non-negative real-valued function. Both forms have their own pros and cons and neither dominates the other. Whereas the former is easier to elicit from agents, the latter is far more expressive, enabling an agent to not only report if she prefers good $A$ to good $B$ but also by how much. \cite{Abdulkadirouglu-Cardinal} exploit this greater expressivity of cardinal utilities to give mechanisms for school choice which are superior to ordinal-utility-based mechanisms. 

Example \ref{ex.GTV}, taken from \cite{Garg-ADHZ}, provides a very vivid illustration of the advantage of cardinal utilities over ordinal ones in one-sided matching markets. 

\begin{example}[Cardinal vs Ordinal Utilities]\label{ex.GTV}
Consider an instance with three types of goods, $T_1, T_2, T_3$, which are present in the proportion of $(1\%, \ 97\%, \ 2\%)$. Based on their utility functions, the agents are partitioned into two sets $A_1$ and $A_2$, where $A_1$ constitute $1\%$ of the agents and $A_2$, $99\%$. The utility functions of agents in $A_1$ and $A_2$ for the three types of goods are $(1, \ \epsilon, \ 0)$ and $(1, \ 1- \epsilon, \ 0)$, respectively, for a small number $\epsilon > 0$. The main point is that whereas  agents in $A_2$ marginally prefer $T_1$ to $T_2$, those in $A_1$ overwhelmingly prefer $T_1$ to $T_2$. 

Clearly, the ordinal utilities of all agents in $A_1 \cup A_2$ are the same. Therefore, a mechanism based on such utilities will not be able to make a distinction between the two types of agents. On the other hand, the HZ mechanism, which uses cardinal utilities, will fix the price of goods in $T_3$ to be zero and those in $T_1$ and $T_2$ appropriately so that by-and-large the bundles of $A_1$ and $A_2$ consist of goods from $T_1$ and $T_2$, respectively.
\end{example}

\section{Preliminaries}
\label{sec.preliminaries}

\subsection{The Nash Bargaining Game} 
\label{sec.Nash}

An {\em $n$-person Nash bargaining game} consists of a pair $(\CN, \cb)$, where $\CN \subseteq \R_+^n$ is a compact, convex set called the {\em feasible set} -- its elements are vectors whose components are utilities that the $n$ players can simultaneously accrue. Point $\cb\in \CN$ is the {\em disagreement point} -- its components are utilities which the $n$ players accrue if they decide not to participate in the proposed solution. 

The set of $n$ agents will be denoted by $A$ and the agents will be numbered $1, 2, \dots, n$. Instance $(\CN, \cb)$ is said to be {\em feasible} if there is a point in $\CN$ at which each agent does strictly better than her disagreement utility, i.e., $\exists \vv \in \CN$ such that $\forall i \in A, \ v_i > c_i$, and {\em infeasible} otherwise. In game theory it is customary to assume that the given Nash bargaining problem $(\CN, \cb)$ is feasible; we will make this assumption as well.

The solution to a feasible instance is the point $\vv \in \CN$ that satisfies the following four axioms:

\begin{enumerate}
\item
{\em  Pareto optimality:}  No point in $\CN$ weakly dominates $\vv$.
\item
{\em  Symmetry:} If the players are renumbered, then a corresponding renumber the coordinates of $\vv$ is a solution to the new instance.
\item
{\em  Invariance under affine transformations of utilities:} If the utilities of any player are redefined by
multiplying by a scalar and adding a constant, then the solution to the transformed problem is obtained by
applying these operations to the particular coordinate of $\vv$.
\item
{\em  Independence of irrelevant alternatives:} If $\vv$ is the solution to $(\CN, \cb)$, and 
$\CS \subseteq \R_+^n$ is a compact, convex set satisfying $c \in \CS$ and  $\vv \in \CS \subseteq \CN$, then $\vv$ is also the solution to $(\CS, \cb)$.
\end{enumerate}

Via an elegant proof, Nash proved:

\begin{theorem}[\citealp{Nash1953two}]
\label{thm.nash}
If the game $(\CN, \cb)$ is feasible then there is a unique point in $\CN$ satisfying the axioms stated above, which is obtained by maximizing $\Pi_{i \in A}  {(v_i - c_i)}$ over $\vv \in \CN$.
\end{theorem}

Nash's solution to his bargaining game involves maximizing a concave function over a convex domain, and is therefore the optimal solution to the following convex program.

	\begin{maxi}
		{} {\sum\nolimits_{i \in A}  {\log (v_i - c_i)}}
			{\label{eq.CP-Nash}}
		{}
		\addConstraint{}{\vv \in \CN}
	\end{maxi}

As a consequence, if for a specific game, a separation oracle can be implemented
in polynomial time, then using the ellipsoid algorithm one can get as good an approximation to the solution of this convex program as desired in time polynomial in the number of bits of accuracy needed \citep{GLS, Vishnoi.book}.

\subsection{Fisher Market Model}
\label{sec.Fisher}

The {\em Fisher market model} consists of a set $A = \{1, 2, \ldots n\}$ of agents and a set $G = \{1, 2, \ldots, m\}$ of infinitely divisible goods. By fixing the units for each good, we may assume without loss of generality that there is a unit of each good in the market. Each agent $i$ has money $m_i \in \Qplus$.  

Let $x_{ij}, \ 1 \leq j \leq m$ represent a {\em bundle of goods allocated to agent $i$}. Each agent $i$ has a utility function $u: \R_+^m \rightarrow \Rplus$ giving the utility accrued by $i$ from a bundle of goods. We will assume that $u$ is concave and weakly monotonic. Each good $j$ is assigned a non-negative price, $p_j$. Allocations and prices, $\xx$ and $\pp$, are said to form an {\em equilibrium} if each agent obtains a utility maximizing bundle of goods at prices $\pp$ and the {\em market clears}, i.e., each good is fully sold to the extent of one unit and all money of agents is fully spent. We will assume that each agent derives positive utility from some good and for each agent, there is a good which gives her positive utility; clearly, otherwise we may remove that agent or good from consideration.

\subsection{Arrow-Debreu Market Model}
\label{sec.AD}

The Arrow-Debreu market model, also known as the {\em exchange model} differs from Fisher's model in that agents come to the market with initial endowments of good instead of money. The union of all goods in initial endowments are all the goods in the market. Once again, by redefining the units of each good, we may assume that there is a total of one unit of each good in the market. The utility functions of agents are as before. The problem now is to find non-negative prices for all goods so that if each agent sells her initial endowment and buys an optimal bundle of goods, the market clears. Clearly, if $\pp$ is equilibrium prices then so is any scaling of $\pp$ by a positive factor.

\subsection{Hylland-Zeckhauser Scheme}
\label{sec.HZ}

Let $A = \{1, 2, \ldots n\}$ be a set of $n$ agents and $G = \{1, 2, \ldots, n\}$ be a set of $n$ indivisible goods. The goal of the HZ scheme is to allocate exactly one good to each agent. However, in order to use the power of a pricing mechanism, which endows the HZ scheme with the properties of Pareto optimality and incentive compatibility in the large, it casts this one-sided matching market in the mold of a linear Fisher market as follows. 

Goods are rendered  divisible by assuming that there is one unit of probability share of each good, and utilities $\{u_{ij}\}$ are defined as in a linear Fisher market. Let $x_{ij}$ be the allocation of probability share that agent $i$ receives of good $j$. Then, $\sum_j {u_{ij} x_{ij}}$ is the {\em expected utility} accrued by agent $i$. Each agent has 1 dollar for buying these probability shares and each good $j$ has a price $p_j \geq 0$.

Beyond a Fisher market, an additional constraint is that the total probability share allocated to each agent is one unit, i.e., the entire allocation must form a {\em fractional perfect matching} in the complete bipartite graph over vertex sets $A$ and $G$. Subject to these constraints, each agent buys a utility maximizing bundle of goods. Another point of departure from a linear Fisher market is that in general, an agent's optimal bundle may cost less than one dollar, i.e., the agents are not required to spend all their money. Since each good is fully sold, the market clears. Hence these are defined to be {\em equilibrium allocation and prices}.

Clearly, an equilibrium allocation can be viewed as a doubly stochastic matrix. The Birkhoff-von Neumann procedure then extracts a random underlying perfect matching in such a way that the expected utility accrued to each agent from the integral perfect matching is the same as from the fractional perfect matching. Since {\em ex ante} Pareto optimality implies {\em ex post} Pareto optimality, the integral allocation will also be Pareto optimal.

\section{Nash-Bargaining-Based Models}
\label{sec.models}

\subsection{One-Sided Matching Markets}
\label{sec.1-models}

We will define five one-sided matching market models based on our Nash bargaining approach. For each model, we will also give a standard  application. For the case of linear utilities, we have singled out the Fisher and Arrow-Debreu versions, namely {\em 1LF} and {\em 1LAD}, since we will study both in some detail later in the paper. For more general utility functions we have defined only the Arrow-Debreu version; the Fisher version is obtained by setting all disagreement utilities to zero. 

It is  easy to see that the fourth one generalizes the first three; however, the earlier ones involve less notation and have an independent standing of their own, hence necessitating all four definitions. The fifth market has a different character. It involves two distinct types of goods; each agent wishes to get one unit of each type. We have assumed linear utilities for this market; generalizing to other utility functions is straightforward using the definitions of the previous four markets. 

Our one-sided matching market models consist of a set $A = \{1, 2, \ldots n\}$ of agents and a set $G = \{1, 2, \ldots, n\}$ of infinitely divisible goods; observe that there is an equal number of agents and goods. There is one unit of each good and each agent needs to be allocated a total of one unit of goods. Hence the allocation needs to be a fractional perfect matching, as defined next.

\begin{definition}\label{def.pm}
Let us name the coordinates of a vector $\xx \in \R_+^{n^2}$ by pairs $i, j$ for $i \in A$ and $j \in G$. Then $\xx$ is said to be a {\em fractional perfect matching} if
	\begin{align}
	    \forall i \in A: \ \sum\nolimits_{j\in G} {x_{ij}} = 1 \ \ \ \mbox{and} \ \ \ \forall j \in G: \ \sum\nolimits_{i\in A} {x_{ij}} = 1. \label{eq:matching_polytope}
	\end{align}
We denote the set of fractional perfect matchings by $\calX$.
\end{definition}

As mentioned in Section \ref{sec.preliminaries}, an equilibrium allocation can be viewed as a doubly stochastic matrix, and the Birkhoff-von Neumann procedure \citep{Birkhoff1946tres, von1953certain} can be used to extract a random underlying perfect matching in such a way that the expected utility accrued to each agent from the integral perfect matching is the same as from the fractional perfect matching. 

Next, we present the matching markets.
In Section \ref{sec.CPs} we prove that each of the matching markets defined below admits a convex program.

\subsubsection{Linear Fisher Nash bargaining one-sided matching market (1LF).} Under {\em 1LF}, each agent $i \in A$ has a linear utility function, as defined in Section \ref{sec.Fisher}. Corresponding to each fractional perfect matching $\xx$, there is a vector $\vv_{\xx}$ in the feasible set $\CN$; its components are the utilities derived by the agents under the allocation given by $\xx$. The disagreement point $\cb$ is the origin. Observe that the setup of {\em 1LF} is identical to that of the HZ mechanism; the difference lies in the definition of the solution to an instance. Its standard application is matching agents to goods.

\subsubsection{Linear Arrow-Debreu Nash bargaining one-sided matching market (1LAD).}
Under {\em 1LAD}, each agent $i \in A$ has a linear utility function, as above. Additionally, we are specified an initial fractional perfect matching $\xx_I$ which gives the initial endowments of the agents. Each agent has one unit of initial endowment over all the goods and the total endowment of each good over all the agents is one unit, as given by $\xx_I$. These two pieces of information define the utility accrued by each agent from her initial endowment; this is her disagreement point $c_i$. As stated in Section \ref{sec.Nash}, we will assume that the problem is feasible, i.e., there is a fractional perfect matching, defining a redistribution of the goods, under which each agent $i$ derives strictly more utility than $c_i$. Each vector $\vv \in \CN$ is as defined in {\em 1LF}. Henceforth, we will consider the slightly more general problem in which we are specified the disagreement point $\cb$ and not the initial endowments $\xx_I$. There is no  guarantee that $\cb$ comes from a valid fractional perfect matching of initial endowments. However, we still want the problem to be feasible. This model is applicable when agents start with an initial endowment of goods and exchange them to improve their happiness.

\subsubsection{Separable piecewise-linear concave Arrow-Debreu Nash bargaining one-sided matching market (1SAD).} This model is analogous to 1LAD, with the difference that each agent has a separable, piecewise-linear concave utility function, hence generalizing the linear utility functions specified in 1LAD. Economists model {\em diminishing marginal utilities} via concave utility functions. Since we are in a fixed-precision model of computation, we have considered separable, piecewise-linear concave (SPLC) utility functions. 
We next define these functions in detail.

For each agent $i$ and good $j$, function $f_{ij}: \Rplus \rightarrow \Rplus$ gives the utility derived by $i$ as a function of the amount of good $j$ she receives. Each $f_{ij}$ is a non-negative, non-decreasing, piecewise-linear, concave function. The overall utility of buyer $i$, $u_i(\xx_i)$, for bundle $\xx_i=(x_{i1},\ldots,x_{in})$ of goods, is additively separable over the goods, i.e., $u_i(\xx_i) = \sum_{j \in G} f_{ij}(x_{ij})$.

We will call each piece of $f_{ij}$ a {\em segment}. 
Number the segments of $f_{ij}$ in order of decreasing slope; throughout we will assume that these segments are indexed by $k$ and that $\kappa_{ij}$ is the number of segement.  
Let $u_{ijk}$ denote the rate at which $i$ accrues utility per unit of good $j$ received, when she is getting an allocation corresponding to this segment.
and let $l_{ijk}$ denote the amount of good $j$ represented by this segment; we will assume that the last segment in each function is of unbounded length and that  that $u_{ijk}$ and $l_{ijk}$ are rational numbers. 
Clearly, the maximum utility she can receive corresponding to this segment is $u_{ijk} \cdot l_{ijk}$. With this definition, we can express $f_{ij}(x_{ij})$ as a piecewise-linear concave function of the form
\begin{align}
    f_{ij}(x_{ij})=\min_{k\in \range{\kappa_{ij}}}\left\{u_{ijk}x_{ij}+b_{ijk}\right\},
\end{align}
where $b_{ijk}$ is the intercept of the $k^{\text{th}}$ segment, which can be recursively expressed as $b_{ij1}=0$ and $b_{ijk}=b_{ij(k-1)}- (u_{ijk}-u_{ij(k-1)})\sum_{h\le k}l_{ijh}$ for $k>1$.


\subsubsection{Non-separable piecewise-linear concave Arrow-Debreu Nash bargaining one-sided matching market (1NAD).} This model differs from {\em 1SAD} in that agents' utility functions are now assumed to be non-separable, piecewise-linear concave. These utility functions are very general and can be used to capture whether goods are complements or substitutes and much more. 
These functions are defined next. 

For each agent $i$, the parameter $\kappa_i$ specifies the number of hyperplanes used for defining the utility of $i$. The latter, $u_i(\xx_i)$, for bundle $\xx_i=(x_{i1},\ldots,x_{in})$ of goods is defined to be
\[ u_i(\xx_i) = \min_{k \in \range{\kappa_i}} \left\{ \sum\nolimits_{j \in G} {a_{ijk} x_{ij} + b_{ik} }\right\} , \]
where $a_{ijk}$ and $b_{ik}$ are non-negative rational numbers. Furthermore, $b_{ik} = 0$ for at least one value of $k$ so that the utility derived by $i$ from the empty bundle is zero. 

Leontief utilities is a fundamental special case of non-separable piecewise-linear concave utilities under which agents want goods in specified ratios. It is used for modeling utilities when goods are complements. In this case, for each agent $i$, we are specified a set $S_i \subseteq G$ of goods she is interested in, and rational numbers $a_{ij} > 0$ for $j\in S_i$, and her utility is characterized as
\begin{align*}
    u_i(\xx_i) = \min_{j \in S_i} \left\{ \frac{x_{ij}}{a_{ij}} \right\}
\end{align*}

Non-separable piecewise-linear concave utility functions are also related to robust sharing problems, where the agents face multiple scenarios and wish to get the highest utility under the worst-case scenario \citep[cf.][]{brown1979knapsack,yu1996max}.

\subsubsection{One-sided matching market over multiple types of goods.}
We turn next to more general settings in which each agent needs to be matched to goods in multiple sets. For instance, suppose we are given two sets of goods, $G_1$ and $G_2$, with $|G_1| = |G_2| = n$, and each agent needs to be matched to one good in each set. If the utility functions of agents are additively separable across $G_1$ and $G_2$, the answer is straightforward, namely independently solve the two matching market problems, $(A, G_1)$ and $(A, G_2)$. 

Next we consider the case of non-separable utility functions. Under the {\em linear Fisher Nash bargaining one-sided matching market over two types of goods}, abbreviated {\em 1LF2G}, the utility of agent $i \in A$ is defined to be 
\[ v_i = \sum_{j \in G_1, l \in G_2} {u_{ijl} x_{ijl}} , \]
where $u_{ijl}$ is the utility accrued by $i$ on obtaining one unit of $j \in G_1$ and one unit of $l \in G_2$. There are numerous natural applications of this problem, e.g., allocation of courses to students from two different majors \citep[see, e.g.,][for other relevant examples.]{carlier2010matching} Akin to {\em 1LF} defined above, one can define generalizations of {\em 1LF2G} in a straightforward manner.


\subsection{Two-Sided Matching Markets}
\label{sec.2-models}

Our two-sided matching market model consist of a set $A = \{1, 2, \ldots n\}$ of workers and a set $J = \{1, 2, \ldots, n\}$ of jobs/firms. For uniformity, we have assumed that there is an equal number of workers and firms, though the model can be easily enhanced and made more general. Our goal is to find an integral perfect matching between workers and firms. In this setting, it is natural to assume that each side has a utility function over the other side, making this a two-sided matching market \citep[cf.][for relevant studies]{ashlagi2022assortment, shi2022optimal}. 

As before, we will relax the problem to finding a fractional perfect matching, $\xx$, followed by rounding as described above. We will explicitly define only the simplest case of two-sided markets; more general models follow along the same lines as one-sided markets. 

Under the {\em linear Fisher Nash bargaining two-sided matching market}, abbreviated {\em 2LF}, the utility accrued by agent $i \in A$ under allocation $\xx$,
\[ u_i(\xx) = \sum\nolimits_{j \in J} {u_{ij} x_{ij} }, \]
where $u_{ij}$ is the utility accrued by $i$ if she were assigned job $j$ integrally. Analogously,
the utility accrued by job $j \in J$ under allocation $\xx$,
\[  w_j(\xx) = \sum\nolimits_{i \in A} {w_{ij} x_{ij} }, \]
where $w_{ij}$ is the utility accrued by $j$ if it were assigned to $i$ integrally. 

In keeping with the axiom of symmetry under Nash bargaining, we will posit that the desires of agents and jobs are equally important and we are led to defining the feasible set in a $2n$ dimensional space, i.e., $\CN \subseteq \R_+^{2n}$. The first $n$ components of feasible point $\vv \in \CN$ represent the utilities derived by the $n$ agents, i.e., $u_i(\xx)$s, and the last $n$ components the utilities derived by the $n$ jobs, i.e., $w_j(\xx)$, under a fractional perfect matching $\xx$. Under {\em 2LF}, the disagreement point is the origin, and we seek the Nash bargaining point. A convex program of {\em 2LF} is given in (\ref{eq.2CP}).

\section{Convex Programs for the Models}
\label{sec.CPs}


\subsection{Primal Formulations}\label{sec.CPs-primal}

We start by presenting convex programs for {\em 1LF} and {\em 1LAD}, namely \eqref{eq.CP-LF} and \eqref{eq.CP-LAD}. These differ only in that the latter has the parameters $c_i$ in the objective function.
%
%
\begin{maxi}
	{} {\sum\nolimits_{i \in A}  {\log \left(\sum\nolimits_{j\in G} {u_{ij} x_{ij}}\right)}}
	{\label{eq.CP-LF}}
	{\text{[1LF]}\quad}
	\addConstraint{\xx}{\in \calX}
\end{maxi}
\begin{maxi}
	{} {\sum\nolimits_{i \in A}  {\log \left(\sum\nolimits_{j\in G} {u_{ij} x_{ij}} - c_i\right)}}
	{\label{eq.CP-LAD}}
	{\text{[1LAD]}\quad}
	\addConstraint{\xx}{\in \calX}
\end{maxi}


Program \eqref{eq.CP-SPLC} is a convex program for {\em 1SAD}.
\begin{maxi}
	{} {\sum\nolimits_{i \in A}  {\log \left(\sum\nolimits_{j\in G} {f_{ij}} - c_i\right)}}
	{\label{eq.CP-SPLC}}
	{\text{[1SAD]}\quad}
    \addConstraint{f_{ij}}{\le u_{ijk} x_{ij} + b_{ijk} \quad }{\forall i \in A, j\in G, k\in \range{\kappa_{ij}}}
	\addConstraint{\xx}{\in \calX}
\end{maxi}


Program \eqref{eq.CP-NPLC} is a convex program for {\em 1NAD}.
\begin{maxi}
    {} {\sum\nolimits_{i \in A}  {\log (v_i - c_i)}}
        {\label{eq.CP-NPLC}}
    {\text{[1NAD]}\quad}
    \addConstraint{v_i}{\leq \sum\nolimits_{j\in G} {a_{ijk} x_{ij} + b_{ik}} \quad }{\forall i \in A, k \in \range{\kappa_i}}
    \addConstraint{\xx}{\in \calX}
\end{maxi}


Program \eqref{eq.CP-2G} is a convex program for {\em 1LF2G}. Note that the feasible region in \eqref{eq.CP-2G} corresponds to a 3-dimensional matching polytope, or more generally, a multi-marginal optimal transport polytope \citep[cf.][]{peyre2019computational}.
\begin{maxi}
    {} {\sum_{i \in A}  {\log \left(\sum_{j\in G_1}\sum_{l\in G_2} {u_{ijl} x_{ijl}}\right)}}
        {\label{eq.CP-2G}}
    {\text{[1LF2G]}\quad}
    \addConstraint{\sum_{j\in G_1}\sum_{l\in G_2} {x_{ijl}}}{=1 \quad}{\forall i \in A}
    \addConstraint{\sum_{i\in A}\sum_{l\in G_2} {x_{ijl}}}{=1 }{\forall j \in G_1}
    \addConstraint{\sum_{i\in A}\sum_{j\in G_1} {x_{ijl}}}{=1 }{\forall l \in G_2}
    \addConstraint{x_{ijl}}{\geq 0}{\forall i \in A, j \in G_1, l \in G_2}
\end{maxi}

Program \eqref{eq.2CP} is a convex program for {\em 2LF}.
\begin{maxi}
    {} {\sum\nolimits_{i \in A}  \log \left(\sum\nolimits_{j\in J} {u_{ij} x_{ij}}\right) \ + \ \sum\nolimits_{j \in J} {\log \left(\sum\nolimits_{i\in A} {w_{ij} x_{ij}}\right)}}
        {\label{eq.2CP}}
    {\text{[2LF]}\quad}
    \addConstraint{\xx}{\in \calX}
\end{maxi}

\subsection{Optimality Bounds}\label{sec.CPs-bounding}
We now describe lower bounding schemes for the optimal utilities using duality theory, which unveil and extend interesting properties of the Nash bargaining based matching markets. In particular, we generalize and improve the equal share matching fairness property of the matching markets with linear utilities, a weakened version of envy-freeness, which states that in a Nash solution, each agent gets a utility which is at least as good as half the utility that she would accrue with goods equally shared among all agents. From a computational perspective, we shall also use these bounds to improve the performance of our solution methods detailed in Section~\ref{sec.SM}. 

In the following, with a slight abuse of notation, we will use $u_i(\yy)$ for $\yy\in \R^n_+$, to denote the utility of agent $i$ for a solution that has $\yy$ as its $i^{\text{th}}$ row.

\subsubsection{One-sided matching markets without disagreement points.}
We first provide lower bounds on the optimal utilities in the one-sided matching markets with general concave utility functions when the disagreement point is the origin.

\begin{lemma}\label{lemma:utility_bound-feas}
    Let $\xx^*$ be an optimal solution to $\max\limits_{\xx \in \calX}\; \sum_{i\in A}\log(u_i(\xx))$, 
    where $u_i(\xx)$ is a concave utility function with $u_i(\vecz)=0$.
    For any $\bar{\xx}\in \R^{n^2}_+$ (not necessarily a perfect matching), we have
    \begin{align}
        u_i(\xx^*)\ge \frac{1}{\sum_{j\in G}\bar{x}_{ij}+n\max_{j\in G}\{\bar{x}_{ij}\}} u_i(\bar{\xx}). \label{eq:concave-util-bound-feas}
    \end{align}
\end{lemma}

Note that Lemma \ref{lemma:utility_bound-feas} does not require $\bar{\xx}$ to be a perfect matching,
thus we can produce the tightest bound for each agent $i$ by maximizing the lower bound in \eqref{eq:concave-util-bound-feas} for $i$ independently of other agents. As a prelude, let us consider the bounds based on highest achievable utilities.

\begin{corollary}\label{corollary:utility_bound-max}
Let $u_i^{\text{max}} = \max_{\yy\in \Delta_n}\{u_i(\yy)\}$ be the highest utility achievable by agent $i$ independently of other agents, where $\Delta_n$ is the unit simplex in $\R^n$. Then
\begin{align}
   \frac{1}{1+n}u_i^{\text{max}}\le u_i(\xx^*) \le u_i^{\text{max}}. \label{eq:concave-util-bound-feas-max}
\end{align}
\end{corollary}
Let us now consider $\bar{\xx}=\theta\hat{\xx}$ for $\theta\in \R_+$ and $\hat{\xx}\in \R^{n^2}_+$ fixed. The bound in \eqref{eq:concave-util-bound-feas} becomes
\begin{align}
    u_i(\xx^*)\ge \frac{1}{\sum_{j\in G}\hat{x}_{ij}+n\max_{j\in G}\{\hat{x}_{ij}\}} \frac{u_i(\theta\hat{\xx})}{\theta}\label{eq:concave-util-bound-feas-theta}
\end{align}
Concavity of $u_{i}$ implies that $\frac{u_i(\theta\hat{\xx})}{\theta}$ increases as $\theta\to 0$ thus \eqref{eq:concave-util-bound-feas-theta} becomes a tighter bound. 
For (piecewise) linear $u_i$ and for sufficiently small $\theta$, $u_i(\theta\hat{\xx})$ is homogeneous of degree 1, hence $u_i(\theta\hat{\xx})=\theta \bar{u}_i(\hat{\xx})$, where $\bar{u}_i$ denotes $u_i$ when only the homogeneous pieces of $u_i$ are considered. 
Given $\sum_{j}\hat{x}_{ij}+n\max_{j}\{\hat{x}_{ij}\}=2n$ for 
$\hat{\xx}=\ee$, where $\ee$ is the vector of all ones, we derive the following result.
\begin{corollary}\label{corollary:utility_bound-equal}
     Let $\bar{u}_i(\xx)$ denote $u_i$ when only the homogeneous pieces of $u_i$ are considered. Then
    \begin{align}
         u_i(\xx^*)\ge \frac{1}{2n}\bar{u}_i(\ee) \label{eq:concave-util-bound-equal}
    \end{align}
\end{corollary}
Clearly, $\bar{u}_i(\xx)=\sum_{j}u_{ij}x_{ij}$ for 1LF, and $\bar{u}_i(\xx)=\sum_{j}u_{ij1}x_{ij}$ for 1SAD, where $u_{ij1}$ is the utility rate in the first piece of the utility function of $j$ for $i$. For 1NAD, assuming that the first $\bar{\kappa}_i$ pieces of $u_{i}$ are homogeneous (i.e., $b_{ik}=0$ for $k\in\range{\bar{\kappa}_i}$), we have $\bar{u}_i(\xx)=\min_{k\in\range{\bar{\kappa}_i}}\sum_{j}u_{ijk}x_{ij}$.

\begin{remark}
    The bound in \eqref{eq:concave-util-bound-equal} is the equal share bound for linear utilities, and is at least as good as the equal share bound for non-linear (e.g., piecewise linear) utilities since $\bar{u}_i(\ee)\ge u_i(\ee)$.
\end{remark}

We can produce even tighter bounds by taking the limits along arbitrary directions. For 1LF, given that $u_i$ is homogeneous of degree 1, we can directly maximize \eqref{eq:concave-util-bound-feas} and obtain
\begin{align}
    \max_{\yy\in \R^n_+} \left\{u_i(\yy): \sum_{j\in G}y_j+n\max_{j\in G}\{y_j\}= 1\right\}
    = \max_{\theta > 0} \max_{\yy\in [0, \theta]^n} \left\{u_i(\yy): \sum_{j\in G}y_j= 1-n\theta\right\}.\label{eq:utility-bound-1lf-knapsack}
\end{align}
Problem \eqref{eq:utility-bound-1lf-knapsack} is essentially a continuous knapsack problem parameterized by $\theta$. 
As given in Proposition \ref{prop:lowerbound-1lf-kp}, we may solve this problem by sorting the goods in the non-decreasing order of their utilities for agent $i$ and adding one good at a time. Consequently, the resulting lower bound generalizes both \eqref{eq:concave-util-bound-feas-max} and \eqref{eq:concave-util-bound-equal} and can be obtained in $\calO(n\log n)$ for each $i$. Figure~\ref{fig:1lf-bound} illustrates the quality of this bound compared to the equal share matching bound \eqref{eq:concave-util-bound-equal} and max-utility bound \eqref{eq:concave-util-bound-feas-max}, in which $\underline{u}_i^{(m)}=\frac{1}{n+m}\sum_{j\in G_m}u_{ij}$ is as defined below and $\underline{u}_i^*=\max\{\underline{u}_i^{(m)}\}$. As illustrated in Figure~\ref{fig:1lf-bound}, it is not difficult to verify that \eqref{eq:utility-bound-1lf-knapsack} is a piecewise linear concave function of $\theta$, with $n$ breakpoints at $\theta_m=\frac{1}{n+m}$ for $m\in\range{n}$, hence Proposition \ref{prop:lowerbound-1lf-kp} yields the optimal lower bound in \eqref{eq:concave-util-bound-feas}.

\begin{figure*}[t]
    \centering
    \includegraphics[clip,width=0.35\textwidth]{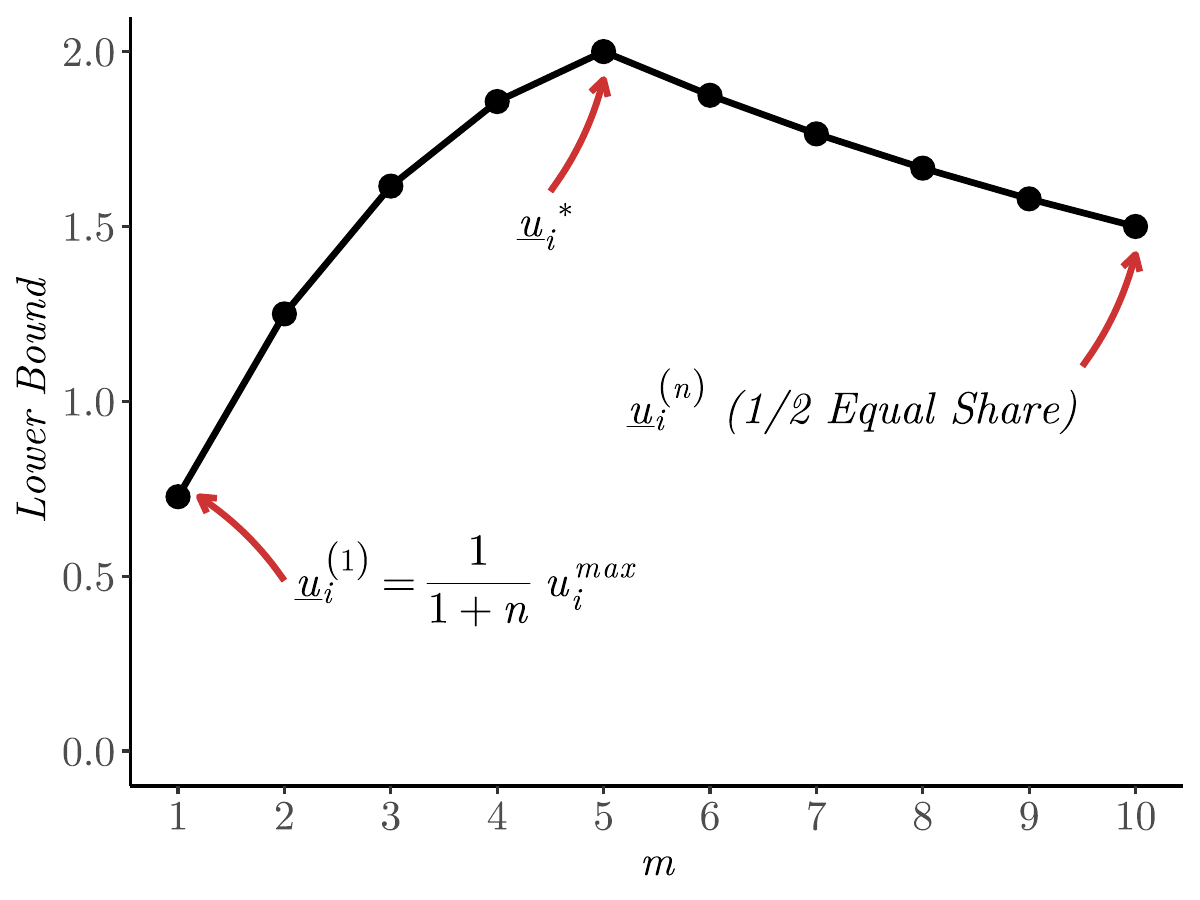}
    \caption{Lower bounds on $u_i(\xx^*)$ in an instance of 1LF with $n=10$ sorted utilities $(8, 7, 6, 5, 4, 0, 0, 0, 0, 0)$.}\label{fig:1lf-bound}
\end{figure*}

\begin{proposition}\label{prop:lowerbound-1lf-kp}
Let $G_m$ denote the indices of the top $m$ utilities in $\{u_{ij}\}_{j\in G}$. Then
\begin{align*}
    u_i(\xx^*)\ge \max_{m\in \range{n}} \left\{\frac{1}{n+m}\sum\nolimits_{j\in G_m}u_{ij}\right\}.
\end{align*}
\end{proposition}

For 1SAD, given that $u_{i}$ is separable piecewise linear, once $y_{j}\le l_{ij1}$ for each $j$, $u_{i}$ becomes linear with $u_{i}(\yy)=\sum_{j\in G}u_{ij1}y_j$. Thus we can maximize \eqref{eq:concave-util-bound-feas} as in the linear case.
\begin{corollary}\label{corollary:lowerbound-1sad-max}
Let $v_{ij}=u_{ij1}$ be the utility rate of agent $i$ in the first piece of $u_{i}$ for good $j$ in 1SAD, and $G_m$ be the indices of the top $m$ utility rates in $\{v_{ij}\}_{j\in G}$. Then $u_i(\xx^*)\ge \max\limits_{m\in \range{n}} \left\{\frac{1}{n+m}\sum_{j\in G_m}v_{ij}\right\}$.
\end{corollary}

For 1NAD, again assuming that the first $\bar{\kappa}_i$ pieces of $u_{i}$ are homogeneous (i.e., $b_{ik}=0$ for $k\in\range{\bar{\kappa}_i}$), for sufficiently small $\yy$ we have $u_i(\yy)=\min_{k\in\range{\bar{\kappa}_i}}\sum_{j}a_{ijk}y_j$, and \eqref{eq:utility-bound-1lf-knapsack} becomes a parameterized \textit{max-min knapsack problem} \citep{brown1979knapsack,yu1996max}. 
By LP duality, there exists $\llambda\in\Delta_{\bar{\kappa}_i}$ such that $u_i(\yy)=\sum_{k\in\range{\bar{\kappa}_i}}\lambda_k\sum_{j}a_{ijk}y_j$ in an optimal solution to \eqref{eq:utility-bound-1lf-knapsack}. For simplicity, we use $\lambda_k=1/\bar{\kappa}_i$, i.e. the average utility rate of $(i,j)$ among the homogeneous piece of $u_{i}$. 


\begin{corollary}\label{corollary:lowerbound-1nad-kp}
In 1NAD, let the first $\bar{\kappa}_i$ pieces of $u_{i}$ be homogeneous.
Let $v_{ij}=\frac{1}{\bar{\kappa}_i}\sum\nolimits_{k\in\range{\bar{\kappa}_{i}}}\{a_{ijk}\}$, $G_m$ denote the indices of the top $m$ values in $\{v_{ij}\}_{j\in G}$, and $\bar{m}=\argmax\limits_{m\in \range{n}} \left\{\frac{1}{n+m}\sum_{j\in G_m}v_{ij}\right\}$. Then 
\begin{align*}
    u_i(\xx^*)\ge \frac{1}{n+\bar{m}}\min_{k\in\range{\bar{\kappa}_i}}\left\{ \sum\nolimits_{j\in G_{\bar{m}}}a_{ijk}\right\}.
\end{align*}
\end{corollary}

\subsubsection{One-sided matching markets with disagreement points.}
The following Lemma extends Lemma~\ref{lemma:utility_bound-feas} to one-sided matching markets with concave utilities and non-zero disagreement points. Note that, unlike Lemma~\ref{lemma:utility_bound-feas}, we now require $\hat{\xx}$ to be a perfect matching.

\begin{lemma}\label{lemma:utility_bound-feas-c}
    Let $\xx^*$ be an optimal solution to $\max_{\xx \in \calX}\; \sum_{i\in A}\log(u_i(\xx)-c_i)$.
    For any $\hat{\xx}\in \calX$ we have
    \begin{align}
        u_i(\xx^*)-c_i \ge \frac{1}{1+\frac{n}{1-\hat{\sigma}\rho}\max_{j}\{\hat{x}_{ij}\}} (u_i(\hat{\xx})-c_i). \label{eq:concave-util-bound-feas-c}
    \end{align}
    where $\hat{\sigma}= \max\limits_{i\in A}\left\{\frac{c_i}{u_i(\hat{\xx})}\right\}$ is a disagreement ratio for $\hat{\xx}$ and  $\rho = \max\limits_{i\in A, j\in G}\left\{\frac{g^{\max}_{ij}}{g^{\min}_{ij}}\right\}$, in which $g^{\max}_{ij}$ and $g^{\min}_{ij}$ are the highest and lowest utility rates for agent $i$ from good $j$ in $u_{i}$, respectively.
\end{lemma}

\begin{remark}
    For $\cb=\vecz$, the bound in \eqref{eq:concave-util-bound-feas-c} is the same as \eqref{eq:concave-util-bound-feas} since $\hat{\sigma}=0$ when $\cb=\vecz$. Additionally, $\rho=1$ for linear utilities, thus \eqref{eq:concave-util-bound-feas-c} recovers the equal share bound for 1LF by setting $\hat{\xx}=\frac{1}{n}\ee$.
\end{remark}

\subsubsection{One-sided matching markets over multiple goods.}\label{sec:bound-1lf2g}
The following result extends Lemma \ref{lemma:utility_bound-feas} to the one-sided matching markets over $m$ goods (without disagreement points).
\begin{lemma}\label{lemma:utility_bound-multiple}
    Let $\xx^*$ be an optimal solution to a one-sided matching market over $m$ goods with $u_i(\xx)$ a concave utility function such that $u_i(\vecz)=0$. Then
    \begin{align}
        u_i(\xx^*)\ge \frac{1}{2n^{m}}\max_{\theta >0}\left\{\frac{u_i(\theta\ee)}{\theta}\right\}, \label{eq:concave-util-bound-multiple}
    \end{align}
    where $\ee$ is the vector of all ones in $\R^{n^{m+1}}$. (Note: $\theta\ee$ need not be feasible.)
\end{lemma}

\begin{corollary}\label{colorally:bound_linear-multiple}
    For 1LF2G, bound \eqref{eq:concave-util-bound-multiple} is the equal share matching fairness bound:
    \begin{align*}
        u_i(\xx^*)\ge \frac{1}{2n^2}\sum_{j\in G_1}\sum_{l\in G_2}u_{ijl}.
    \end{align*}
\end{corollary}

\subsubsection{Two-sided matching markets with linear utilities.}
We now describe lower bounds for two-sided matching markets with linear utilities $u_i(\xx)=\sum\limits_{j\in J}u_{ij}x_{ij}$ and $w_j(\xx)=\sum\limits_{i\in A}w_{ij}x_{ij}$. For brevity we describe the bounds for $u_i(\cdot)$; bounds for $w_j(\cdot)$ follow by symmetry. 

\begin{lemma}\label{lemma:bound-2lf}
Let $\xx^*\in \calX$ be an optimal solution to $\max\limits_{\xx \in \calX}\; \sum\nolimits_{i\in A}\log(u_i(\xx))+\sum\nolimits_{j\in J}\log(w_j(\xx))$, and define $\bar{w}_{ij}=\frac{w_{ij}}{\max\limits_{i'\in A}\{w_{i'j}\}}$ and $\bbar{w}_{i}=\min\{n, \max\limits_{j\in J, i'\in A}\{\frac{w_{ij}}{w_{i'j}}\}\}$. 
 Then, for any $\bar{\xx}\in\R^{n^2}_+$ and $i\in A$:

\begin{align}
    u_i(\xx^*)\ge \max\left\{\frac{u_i(\bar{\xx})}{\sum\nolimits_{j\in J}\bar{x}_{ij}(2n-\bar{w}_{ij})},\frac{u_i(\bar{\xx})}{\sum\nolimits_{j\in J}\bar{x}_{ij}(1+\bbar{w}_i-\bar{w}_{ij})+2n\max\nolimits_{j\in J}\{\bar{x}_{ij}\}}\right\}. \label{eq:bound-2lf-u-2}
\end{align}

\end{lemma}
\begin{remark}
    The constants $\bar{w}_{ij}\le 1\le \bbar{w}_i$ capture variations in the utility rates of jobs, with more uniform utilities driving both $\bar{w}_{ij}$ and $\bbar{w}_i$ towards 1. Indeed, when utilities for jobs are completely uniform (i.e., $\bbar{w}_i=1$ and $\bar{w}_{ij}=1$ for each $j$), \eqref{eq:bound-2lf-u-2} recovers the bound \eqref{eq:concave-util-bound-feas} for 1LF with $2n\max\limits_{j\in J}\{\bar{x}_{ij}\}$ in the denominator in place of $n\max\limits_{j\in J}\{\bar{x}_{ij}\}$.
\end{remark}
\begin{remark}
    With $\bar{\xx}=\ee$, \eqref{eq:bound-2lf-u-2} yields a bound slightly worse than $2/3$ the equal share bound:
    \begin{align*}
        u_i(\xx^*)\ge \frac{\sum_{j\in J}u_{ij}}{3n+\sum_{j\in J}(\bbar{w}_i-\bar{w}_{ij})}.
    \end{align*}
\end{remark}

Analogous to one-sided models, we can derive the tightest bound on $u_i(\xx^*)$ by maximizing the right-hand-side values in \eqref{eq:bound-2lf-u-2} independently for each $i$. With the first bound in \eqref{eq:bound-2lf-u-2} we obtain
\begin{align*}
    u_i(\xx^*)\ge & \max_{\bar{\xx}\in\R^{n^2}_+}\left\{\frac{u_i(\bar{\xx})}{\sum\nolimits_{j\in J}\bar{x}_{ij}(2n-\bar{w}_{ij})}\right\} = \max_{j\in J}\left\{\frac{u_{ij}}{2n-\bar{w}_{ij}}\right\}.
\end{align*}
Similar to the one-sided models, since $u_{i}$ is homogeneous, producing the tightest bound based on the second bound in \eqref{eq:bound-2lf-u-2} amounts to solving a parametric knapsack problem:
\begin{align}
  \max_{\theta > 0} \max_{\yy\in [0,\theta]^n}\left\{\sum\nolimits_{j\in J} u_{ij} y_{j}: \sum\nolimits_{j\in J} (1+\bbar{w}_i-\bar{w}_{ij})y_j= 1-2n\theta\right\}. \label{eq:bound-2lf-parametric-kp}
\end{align}
Therefore, using Dantzig's procedure, we can solve the knapsack problems in \eqref{eq:bound-2lf-parametric-kp} by picking one $j$ at a time in the non-decreasing order of adjusted utility rates $v_{ij}=\frac{u_{ij}}{1+\bbar{w}_i-\bar{w}_{ij}}$. In fact, similar to Proposition~\ref{prop:lowerbound-1lf-kp}, we can solve \eqref{eq:bound-2lf-parametric-kp} in $\calO(n\log n)$ by simply sorting the adjusted utilities once.


\begin{proposition}\label{prop:lowerbound-2lf-kp}
Let $J_m$ denote the indices of the top $m$ values in $\{\frac{u_{ij}}{1+\bbar{w}_i-\bar{w}_{ij}}\}_{j\in J}$. Then
\begin{align*}
    u_i(\xx^*)\ge \max_{m\in \range{n}} \left\{\frac{1}{2n+\sum_{j\in J_m}(1+\bbar{w}_i-\bar{w}_{ij})}\sum\nolimits_{j\in J_m}u_{ij}\right\}.
\end{align*}
\end{proposition}
\section{Solution Methods}\label{sec.SM}
We present two solution methods for solving instances of the convex programs given in Section~\ref{sec.CPs}: (a) Frank-Wolfe algorithm, and (b) Cutting-Plane algorithm. Both algorithms rely on linear approximations of these problems and converge to the optimal solution in polynomial time.
For simplicity of exposition, we focus on the simpler models {\em 1LAD} (and {\em 1LF}) and {\em 2LF} to describe the algorithms. We will explain how these algorithms can be extended to other models.

\subsection{Frank-Wolfe Algorithm}\label{sec:fw}
Frank-Wolfe (FW) method \citep{frank1956algorithm, jaggi2013revisiting} is one of the simplest and earliest known iterative algorithms for solving non-linear convex optimization problems of the form 
\begin{align}
    \max_{\xx\in \calX}f(\xx),\label{eq:basic}
\end{align}
where $f$ is a concave function and $\calX$ is a compact convex set.
The underlying principle in Frank-Wolfe method is to replace the non-linear objective function $f$ with its linear approximation $\tilde{f}(\xx)=f(\xx^{(0)})+\nabla f(\xx^{(0)})^{\top}(\xx-\xx^{(0)})$ at a trial point $\xx^{(0)}\in \calX$, and solve a simpler problem $\max_{\xx\in \calX}\tilde{f}(\xx)$
to produce an ``atom'' solution $\hat{\xx}$. The algorithm then iterates by performing line search between $\xx^{(0)}$ and $\hat{\xx}$ to produce the next trial point $\xx^{(1)}$ as a convex combination of $\xx^{(0)}$ and $\hat{\xx}$.

\subsubsection{Linear utilities.}
Algorithm~\ref{pseudo-code-frank-wolfe} presents an implementation of the FW algorithm for solving instances of {\em 1LAD}, in which the objective function $f$ is defined as
$f(\xx)=\sum_{i\in A}\log(\sum_{j\in G}u_{ij}x_{ij}-c_i)$ and the feasible region is the matching polytope defined in Definition~\ref{def.pm}. 

\paragraph{Producing an atom.} FW is particularly useful when one can solve the subproblems in a combinatorial fashion. Given that $\calX$ corresponds to a matching polytope, at each iteration of Algorithm~\ref{pseudo-code-frank-wolfe}, the \textit{atom} is an integral perfect matching produced by solving a maximum weight matching problem, which can be done in $\calO(n^3)$ using, e.g., the Auction algorithm \citep{bertsekas1988auction}.
The optimal solution produced by FW is therefore a convex combination of these integral perfect matchings.

\begin{algorithm}[t!]
	Set $t\leftarrow 0$ and find a feasible perfect matching $\xx^{(0)}$.
 
	\While{not converged}{
		Compute $g_{ij}=\frac{\partial}{\partial x_{ij}} f(\xx^{(t)})=\frac{u_{ij}}{v^{(t)}_i-c_i}$, where $v^{(t)}_i=\sum_{j\in G}u_{ij}x^{(t)}_{ij}$
  
		Compute $\hat{\xx}^{(t)}$ by solving the following maximum weight matching problem:
        \begin{align*}
            \hat{\xx}^{(t)}=\arg\max_{\xx\in \calX} \; \sum\nolimits_{i\in A}\sum\nolimits_{j\in G}g_{ij}x_{ij}
        \end{align*}
		
		Set the step-size $\tau^{(t)}=\frac{2}{t+2}$, or compute $\tau^{(t)}$ using the following line search 
		\begin{align*}
		    \tau^{(t)}=\arg\max_{\gamma \in [0,1]} \; f\left((1-\gamma)\xx^{(t)}+\gamma \hat{\xx}^{(t)}\right)
		\end{align*}
		
		Update $\xx^{(t+1)}=(1-\tau^{(t)})\xx^{(t)}+\tau^{(t)} \hat{\xx}^{(t)}$; $t\leftarrow t+1$
	}
	\caption{Frank-Wolfe algorithm for solving {\em 1LAD}}
	\label{pseudo-code-frank-wolfe}
\end{algorithm}

\paragraph{Step-size.}
For linear utilities, given that $v_i$ is a linear function of $\xx$, instead of performing line search step in Algorithm~\ref{pseudo-code-frank-wolfe} in the space of $\xx$, we can perform line search in the space of utilities: 
\begin{align}
	\tau^{(t)}=\arg\max\limits_{\tau \in [0,1]}\; \sum_{i\in A} \log\left(\tau (\hat{v}^{(t)}_i-v^{(t)}_i)+v^{(t)}_i-c_i\right),\label{eq:line-search}
\end{align}
where $\hat{v}^{(t)}_i=\sum_{j\in G}u_{ij}\hat{x}^{(t)}_{ij}$, which can be solved efficiently using a first-order algorithm such as Newton's method. Consequently, given the polynomial form of the derivatives in \eqref{eq:line-search}, $\tau^{(t)}$ will be  a rational value provided the search is started with a rational value. By induction, we note the following property of FW, which is particularly interesting when  the algorithm terminates in few iterations, yielding a sparse rational convex combination of integral perfect matchings.

\begin{remark}
	\label{ramark.fw-rational}
	Provided $\xx^{(0)}$ is an integral perfect matching, the iterate $\xx^{(t+1)}$ at the end of iteration $t$ of the FW Algorithm~\ref{pseudo-code-cutting-plane} is a rational convex combination of $t+2$ integral perfect matchings.
\end{remark}

\paragraph{Convergence.} 
The standard implementation of the FW algorithm assumes a sublinear convergence rate when gradient of $f$ is Lipschitz continuous \citep{jaggi2013revisiting}; that is, $\calO(\frac{1}{\epsilon})$ iterations are required to achieve an $\epsilon$-optimal solution. Since the logarithmic form of $f$ lacks gradient Lipschitz continuity, recovering the $\calO(\frac{1}{\epsilon})$ rate warrants a more careful treatment \citep[cf.][]{zhao2022analysis}. However, 
we can recover gradient Lipschitz continuity by assuming a lower bound on the optimal utilities (Section~\ref{sec.CPs-bounding}), and thus avoid unboundedness in the gradients as remarked below.
\begin{remark}\label{remark:log_lipschitz}
    Given $v_0>0$, $\log_+(v|v_0)$ defined below has $\frac{1}{v_0^2}$-Lipschitz continuous gradient.
    \begin{align}
    \log_+(v|v_0) = \begin{cases}
        \log(v_0)-1+\frac{v}{v_0} & \text{if } v \le v_0\\
        \log(v) & \text{if } v \ge v_0
    \end{cases}\label{eq:lipschitz-log}
\end{align}
\end{remark}

We also note that, by concavity of $f$,  $\sum_{i,j}g_{ij}(\hat{x}^{(t)}_{ij}-x^{(t)}_{ij})$ provides an upper bound on the optimality gap of iterate $x^{(t)}$ at each iteration $t$ of Algorithm~\ref{pseudo-code-frank-wolfe} \citep{jaggi2013revisiting}. Thus we may numerically assess convergence of Algorithm~\ref{pseudo-code-frank-wolfe} using
\begin{align}
	\text{Gap}=\frac{\sum_{i,j}g_{ij}(\hat{x}^{(t)}_{ij}-x^{(t)}_{ij})}{|f(\xx^{(t)})|}.\label{eq:gap-fw}
\end{align}

\paragraph{Extension to other models.}
The FW Algorithm~\ref{pseudo-code-frank-wolfe} easily extends to {\em 2LF}, with the only difference that the gradient is now 
$$g_{ij}=\frac{\partial}{\partial x_{ij}} f(\xx^{(t)})=\frac{u_{ij}}{v^{(t)}_i}+\frac{w_{ij}}{v^{(t)}_j}.$$
For {\em 1SAD} and {\em 1NAD}, given the non-smooth objective functions, the FW algorithm is not guaranteed to converge to the optimal solution \citep[cf.][for a counter example]{nesterov2018complexity}. Hence, we introduce smooth counterparts for these utility functions in Section~\ref{sec:smooth-fw}. Finally, given that the feasible region of {\em 1LF2G} does not conform to a standard matching polytope, we provide a tailored algorithm in Section~\ref{sec:1lf2g-fw} for the general matching market models over multiple goods.

\subsubsection{Extension to 1NAD and 1SAD.}\label{sec:smooth-fw}
We now introduce smooth counterparts for the piecewise linear utility functions in {\em 1SAD} and {\em 1NAD} based on Moreau envelopes of these functions. 

Given a convex function $\varphi$, the Moreau envelope of $\varphi$ with parameter $\gamma>0$ is defined as 
\begin{align}
    e_{\gamma;\varphi}(\xx) = \inf_{\yy}\left\{\varphi(\yy)+\frac{1}{2\gamma}\|\xx-\yy\|^2\right\}.   \label{eq:Moreau-env}
\end{align}
The unique optimal solutions to \eqref{eq:Moreau-env} is denoted $\prx_{\gamma;\varphi}(\xx)$ and called the \textit{proximal operator}. Remarkably, $e_{\gamma;\varphi}$ is smooth and assumes the same set of minimizers as $\varphi$, and $\nabla e_{\gamma;\varphi}$ given by \eqref{eq:Moreau-env-gradient} is $\frac{1}{\gamma}$-Lipschitz continuous for convex $\varphi$ \citep{parikh2014proximal}.
\begin{align}
    \nabla e_{\gamma;\varphi}(\xx) = \frac{1}{\gamma}\left(\xx-\prx_{\gamma;\varphi}(\xx)\right). \label{eq:Moreau-env-gradient}
\end{align}

\begin{definition}
    For a non-smooth concave function $u$, we denote by $u_{\gamma}(\cdot)=-e_{\gamma;-u}(\cdot)$ the smooth concave envelope of $u$.
\end{definition}

Figure \ref{fig:concave_envelope} illustrates the smooth concave envelopes a piecewise linear function for different choices of $\gamma$. Next, we show how to derive the smooth concave envelope of utility functions in 1NAD.

\begin{figure}[t!]
    \centering
    \includegraphics[width = 0.4\textwidth]{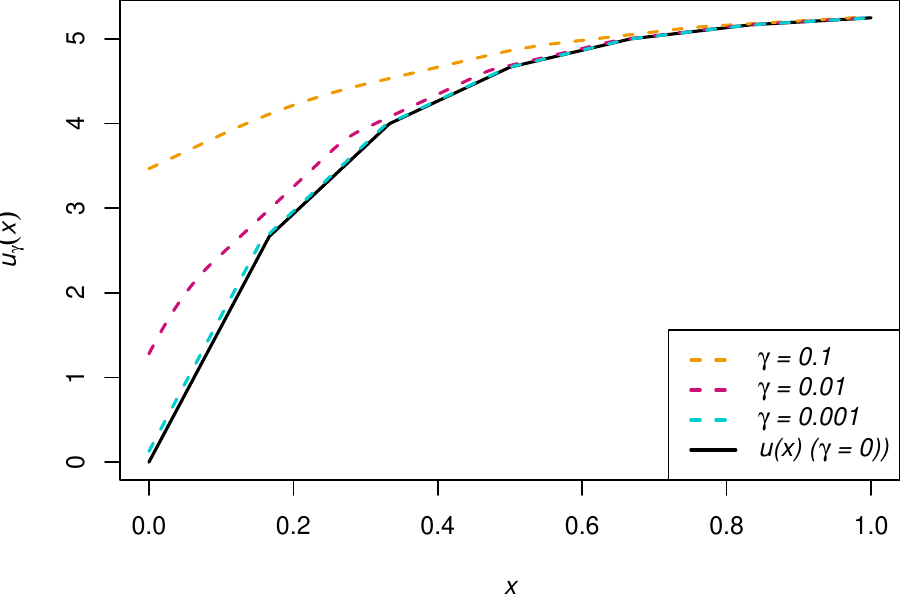}
    \caption{Smooth concave envelopes of a piecewise linear concave utility function.}
    \label{fig:concave_envelope}
\end{figure}

%

\begin{proposition}\label{prop:smooth-nonseparable}
    Let $u_i(\xx_i)=\min_{k\in \range{\kappa_i}}\left\{\sum_{j\in G}a_{ijk} x_{ij} + b_{ik}\right\}$ be a non-separable piecewise linear concave utility function and $u_{i\gamma}(\xx_i)=-e_{\gamma;-u_i}(\xx_i)$ its smooth concave envelope. Then
    \begin{align}
        u_{i\gamma}(\xx_i) = & \min_{\zz\in \Delta_{\kappa_i}}\left\{\frac{\gamma}{2}\zz^{\top}Q_i\zz+\xx_i^{\top}A_i \zz\right\}, \label{eq:Moreau-env-1nad}\\
        \nabla u_{i\gamma}(\xx) = & A_i\zz^*, \label{eq:Moreau-env-1nad-grad}
    \end{align}
    where $\Delta_{\kappa_i}=\{\zz\in \R_+^{\kappa_i}: \sum_{k\in \range{\kappa_i}}z_k=1\}$ is the unit simplex, $A_i$ is an $n\times \kappa_i$ matrix with entries $A_{ijk}=a_{ijk}+b_{ik}$, $Q_i=A_i^{\top}A_i$, and $\zz^*\in \Delta_{\kappa_i}$ is the optimal solution to \eqref{eq:Moreau-env-1nad}.
\end{proposition}
 
\begin{remark}
    The matrix $Q_i$ is positive semi definite, does not depend on $\xx$, and needs to be computed only once for each $i$ in the course of the FW algorithm.
\end{remark}

Consequently, the smooth concave envelope $u_{i\gamma}(\xx_i)$ and its gradient can be obtained as solution to a convex quadratic program over a unit simplex in $\R^{\kappa_i}$, which can be solved extremely fast, using e.g., exponentiated gradient descent or proximal gradient method.


The case of separable utility functions in 1SAD exhibits an even simpler structure, since we can characterize the break points of the one-dimensional piecewise linear functions explicitly. 
Given that the $k^{\text{th}}$ segment of the utility function $f_{ij}$ is of length $l_{ijk}$, the $k^{\text{th}}$ break point of $f_{ij}$ is $\hat{x}_{ijk}=\sum_{h=1}^{k}{l_{ijh}}$. 
By concavity of $f_{ij}$ we have $u_{ij,k-1}\ge u_{ijk}$, and, by definition of the $\hat{x}_{ijk}$ and $b_{ijk}$ values, we have $u_{ijk}\hat{x}_{ij,k-1} + b_{ijk} = u_{ij,k-1}\hat{x}_{ij,k-1} + b_{ij,k-1}$  for each $k>1$.
Thus we can restate $f_{ij}$ as
\begin{align*}
    f_{ij}(x) = \begin{cases}
        u_{ij1}x + b_{ij1} & \text{if } x\in (-\infty, \hat{x}_{ij1}]\\
        u_{ij2}x + b_{ij2} & \text{if } x\in [\hat{x}_{ij1},\hat{x}_{ij2}]\\
        \dots & \dots \\
        u_{ij,\kappa_{ij}}x+ b_{ij,\kappa_{ij}} & \text{if } x\in [\hat{x}_{ij,\kappa_{ij}-1}, \infty)
    \end{cases}
\end{align*}

The following result shows that $f_{ij}$ admits a closed-form concave envelope, in which $ij$ is dropped for simplicity.
%
Note that $f'_{\gamma}(x)$ can be computed by identifying the interval which $x$ belongs to.

\begin{proposition}\label{prop:smooth-separable}
    Let $f$ be a separable concave piecewise linear utility function with $\kappa$ segments and $\kappa-1$ breakpoints $\hat{x}_{1}<\hat{x}_{2}<\dots<\hat{x}_{k-1}$. 
    Then the smooth concave envelope of $f$ is
    \begin{align*}
        f_{\gamma}(x) = \begin{cases}
            u_1 x + b_1 + \frac{\gamma}{2}u_1^2 & \text{if } x\in (-\infty, \hat{x}_1-\gamma u_1]\\
            u_1 \hat{x}_1 + b_1 - \frac{1}{2\gamma}(x-\hat{x}_1)^2 & \text{if } x\in [\hat{x}_1-\gamma u_1,  \hat{x}_1-\gamma u_2]\\
            u_2 x + b_2 + \frac{\gamma}{2}u_2^2 & \text{if } x\in [\hat{x}_1-\gamma u_2, \hat{x}_2-\gamma u_2]\\
            u_2 \hat{x}_2 + b_2 - \frac{1}{2\gamma}(x-\hat{x}_2)^2 & \text{if } x\in [\hat{x}_2-\gamma u_2,  \hat{x}_2-\gamma u_3]\\
            \dots & \dots\\
            u_{\kappa-1} \hat{x}_{\kappa-1} + b_{\kappa-1} - \frac{1}{2\gamma}(x-\hat{x}_{\kappa-1})^2 & \text{if } x\in [\hat{x}_{\kappa-1}-\gamma u_{\kappa-1}, \hat{x}_{\kappa-1}-\gamma u_{\kappa}]\\
            u_{\kappa} x + b_{\kappa} + \frac{\gamma}{2}u_{\kappa}^2 & \text{if } x\in [\hat{x}_{\kappa-1}-\gamma u_{\kappa}, \infty)\\
        \end{cases}
    \end{align*}
\end{proposition}

\subsubsection{Extension to multiple goods.}\label{sec:1lf2g-fw}
Matching markets over $m\ge 2$ goods are the most challenging of our Nash-bargaining-based models due to the $\calO(n^{m+1})$ number of variables involved in these models. In addition, as the feasible region no longer conforms to a matching polytope or any other network-flow polytope, atom solutions cannot be obtained in a combinatorial fashion, curtailing scalability of FW Algorithm~\ref{pseudo-code-frank-wolfe}.
We address this issue through a generalization of the FW algorithm and considering a slightly different convex program instead of \eqref{eq.CP-2G}.
For generality, we describe the algorithm for $m$ goods which covers 1LF2G ($m=2$), 1LF ($m=1$), and beyond. 


Let $\calX^m\subset \R^{n^{m+1}}$ be the $(m+1)$-dimensional matching polytope (e.g., $\calX^1$ is the matching polytope \eqref{eq:matching_polytope} while $\calX^2$ is the set of feasible allocations in 1LF2G).
Let $u_i(\xx)$ be the total utility of agent $i$ from allocation $\xx$ (e.g., $u_i(\xx)=\sum_{j_1\in G_1}\sum_{j_2\in G_2} u_{i,j_1,j_2}x_{i,j_1,j_2}$ in 1LF2G), and $\underline{u}_{i}$ be a lowerbound on the optimal value of $u_i(\xx)$ as described in Section~\ref{sec:bound-1lf2g}. We can replace the objective function with
$f(\xx)=\sum_{i\in A}\log_+(u_i(\xx)|\underline{u}_i)$ which has Lipschitz continuous gradient, and write \eqref{eq.CP-2G} as 
\begin{align}
    \max_{\xx\in \calX^m}f(\xx).\label{eq:basic-1lf2g}
\end{align}
Now, consider the composite optimization problem
\begin{align}
    \max_{\xx\in \calX^m}f(\xx)+\zeta h(\xx),\label{eq:basic-entropic}
\end{align}
where $h(\xx)=-\sum_{i}\sum_{j_1,\dots,j_m}x_{i,j_1,\dots,j_m}(\log(x_{i,j_1,\dots,j_m})-1)$ is the entropy function and $\zeta$ is a sufficiently small positive value. Note that as $\zeta\to 0$, optimal solution of \eqref{eq:basic-entropic} converges to the optimal solution of \eqref{eq:basic-1lf2g}. Moreover, given that $h$ is strongly concave, FW algorithm for solving \eqref{eq:basic-entropic} can achieve a linear convergence rate through away steps \citep{lacoste2015global}. However, this algorithm would still require solving a huge linear program at each iteration $t$:
\begin{align}
    \max_{\xx\in \calX^m}\nabla_{\xx}f(\xx^{(t)})+\zeta \nabla_{\xx}h(\xx^{(t)}).\label{eq:multimarginal-gradient}
\end{align}
The key observation to be made is that the regularized multimarginal matching problem 
\begin{align}
    \max_{\xx\in \calX^m}\nabla_{\xx}f(\xx^{(t)})+\zeta h(\xx)\label{eq:multimarginal-gradient-entropic}
\end{align}
can be solved more efficiently than the linear program \eqref{eq:multimarginal-gradient} thanks to adaptation of Sinkhorn-Knopp's matrix scaling algorithm to multiple margins \citep{peyre2019computational,lin2022complexity}. This property lends itself well to the generalized Frank-Wolfe algorithm for convex programs with composite objective functions, in which one approximates only one of the functions with its linear approximation \citep{nesterov2018complexity}. We describe the generalized FW algorithm for solving instances of 1LF2G in Algorithm \ref{pseudo-code-entropic-frank-wolfe}. With this implementation, not only can we produce an atom more efficiently, but also, as shown by \cite{nesterov2018complexity}, with the \textit{aggressive} choice of step size $\tau_t =  \frac{6(t + 1)}{(t + 2)(2t + 3)}$, we get a linear convergence rate thanks to strong concavity of $h$ and Lipschitz smoothness of $f$.

\begin{algorithm}[h]
	Set $t=0$, and the initial $(m+1)$-dimensional matching $\xx^{(0)}\in \calX^m$ with $x^{(0)}_{i,j_1,\dots,j_m}=\frac{1}{n^m}$.
 
	\While{not converged}{
		Compute the gradient tensor $\gb^{(t)}=\nabla_{\xx} f(\xx^{(t)})$
  
		Find atom $\hat{\xx}^{(t)}$ by solving the regularized $(m+1)$-dimensional  matching problem using Algorithm~\ref{pseudo-code-sinkhorn}:
		\[
		\hat{\xx}^{(t)}=\arg\max_{\xx\in \calX^m} \; \xx^{\top}{\gb}^{(t)} + \zeta h(\xx)
		\]

		Update $\xx^{(t+1)}=(1-\tau^{(t)})\xx^{(t)}+\tau^{(t)} \hat{\xx}^{(t)}$, where $\tau_t =  \frac{6(t + 1)}{(t + 2)(2t + 3)}$, and set $t\leftarrow t+1$
	}
	\caption{Generalized Frank-Wolfe algorithm for matching markets over $m$ goods}
	\label{pseudo-code-entropic-frank-wolfe}
\end{algorithm}

To solve the regularized multimarginal matching problem \eqref{eq:multimarginal-gradient-entropic}, we adapt the Sinkhorn-Knopp's matrix scaling to multiple marginals as described in \cite{lin2022complexity}. For the sake of self-containedness, we describe this procedure in Algorithm~\ref{pseudo-code-sinkhorn} in Appendix~\ref{app:entropic_multimatching}.

\begin{theorem}
    For a fixed value $\zeta$, 
    Algorithm \ref{pseudo-code-entropic-frank-wolfe} converges linearly in the number of iterations to an optimal solution of \eqref{eq:basic-entropic}. Furthermore, each iteration of Algorithm \ref{pseudo-code-entropic-frank-wolfe} takes $\calO\left(\frac{mn^{m+1}\|\gb\|^2_{\infty}}{4\log(n)\zeta^2}\right)$ arithmetic operations.
\end{theorem}



\subsection{Cutting-Plane Algorithm} 
The underlying principle in the cutting-plane method for solving convex programs with nonlinear objective function is to outer-approximate the epigraph of the objective function through a series of linear programs \citep{kelley1960cutting,Vishnoi.book}. Let $f(\vv)=\sum_{i\in A} \log(v_i-c_i)$ be the objective function in {\em 1LAD} where $v_i=\sum_{j\in G}u_{ij}x_{ij}$. Since $f$ is concave in $\vv$, for a given solution $\hat{\vv}$ we have:
\begin{align}
	f(\vv)\le f(\hat{\vv})+ \nabla f(\hat{\vv})^{\top} (\vv-\hat{\vv})=f(\hat{\vv})-n+\sum\nolimits_{i\in A}\frac{v_i-c_i}{\hat{v}_i-c_i}\label{eq:cp-cut}
\end{align}
Therefore, we can rewrite {\em 1LAD} as the following semi-infinite linear program (SILP):


	\begin{maxi}
		{} {\eta}
		{\label{eq.silp}}
		{}
		\addConstraint{\eta\leq f(\hat{\vv}) + \nabla f(\hat{\vv})^{\top} (\vv-\hat{\vv})}{\quad}{\forall \hat{\vv} \in\mathcal{N}}
		\addConstraint{(\xx,\vv)\in \mathcal{S}}{,}{}
	\end{maxi}
	where $\mathcal{N}$ is the set of vectors $\hat{\vv}$ such that $\hat{v}_i>c_i$, and $\mathcal{S}$ is the set of feasible assignments.
%
Observe that replacing $\mathcal{N}$ with $\hat{\mathcal{N}}\subset \mathcal{N}$ in \eqref{eq.silp} yields an LP which is a relaxation of the SILP \eqref{eq.silp}.
A natural way of solving SILP \eqref{eq.silp} is to start with a manageable subset $\hat{\mathcal{N}}$ and grow this set until the upper bound produced by the LP is sufficiently close to the optimal value \citep{kelley1960cutting}. 
Instead of cutting off the corner points of the hypograph approximations, it is customary to solve modified forms of these LPs and cut off interior points. Given $\underline{f}$, a lower bound on the optimal value of $f$,
%
we may construct a cutting plane through the center of the hypograph approximation by solving
\begin{maxi}
	{} {\sigma}
	{\label{eq.central_silp}}
	{}
	\addConstraint{\eta-\sigma}{\geq \underline{f}}{}
	\addConstraint{\eta}{\leq f(\hat{\vv})+ \nabla f(\hat{\vv})^{\top} (\vv-\hat{\vv})-\sigma \|(1,\nabla f(\hat{\vv}))\|_2\quad}{\forall \hat{\vv} \in\hat{\mathcal{N}}}
	\addConstraint{(\xx,\vv)}{\in \mathcal{S},}{}
\end{maxi}
which yields radius $\sigma$ and center $(\vv,\xx,\eta)$ of the largest ball that can be inscribed inside the hypograph approximation \citep{nemhauser1971modified,elzinga1975central}.

Algorithm~\ref{pseudo-code-cutting-plane} describes the proposed \textit{Central Cutting-Plane} (CP) algorithm for solving instances of {\em 1LAD}.
As the algorithm iterates, we improve the lower bound $\underline{f}$ and add new cuts to tighten the hypograph approximation. Consequently, the inscribed ball shrinks (i.e., $\{\sigma^{(t)}\}_{t=0}^{\infty}\to 0$), and $\{(\vv^{(t)},\xx^{(t)})\}_{t=0}^{\infty}$ converges to the optimal solution \citep{elzinga1975central}. 
%
We use the optimality gap in \eqref{eq:gap-cp} and terminate the algorithm once this gap falls below a given optimality gap threshold.
\begin{align}
	\text{Gap}=\frac{\sigma^{(t)}}{|\eta^{(t)}|}.\label{eq:gap-cp}
\end{align}


\begin{algorithm}[t]
		Find an initial solution $(\vv^{(0)},\xx^{(0)})$
  
		Initialize $\hat{\mathcal{N}}\leftarrow\{\vv^{(0)}\}$; $\underline{f} \leftarrow f(\vv^{(0)})$;  $t\leftarrow 1$
  
		$(\vv^*,\xx^*)\leftarrow (\vv^{(0)},\xx^{(0)})$
  
		\While{not converged}{
			Solve LP \eqref{eq.central_silp} to obtain the center $(\vv^{(t)},\xx^{(t)},\eta^{(t)})$ and radius $\sigma^{(t)}$\;
			$\hat{\mathcal{N}}\leftarrow \hat{ \mathcal{N}}\cup\{\vv^{(t)}\}$
   
			\If {$\underline{f}<f(\vv^{(t)})$}{
				$\underline{f}\leftarrow f(\vv^{(t)})$; $(\vv^*,\xx^*)\leftarrow (\vv^{(t)},\xx^{(t)})$}
			$t\leftarrow t+1$
		}
	\caption{Central cutting-plane algorithm for solving {\em 1LAD}}
	\label{pseudo-code-cutting-plane}
\end{algorithm}
\subsubsection{Enhancement techniques.} 

We now make several remarks on enhancing CP. 
\paragraph{Pre-processing.} 
In general, we have a decision variable $x_{ij}$ for each agent $i$ and good $j$. Given that contribution of $x_{ij}$'s with $u_{ij}=0$ to total utility is zero, we can eliminate $x_{ij}$ when $u_{ij}=0$. To recover feasibility, we replace the assignment constrains with the following constraints:
\[ \forall i \in A: \ \sum\nolimits_{j\in G} {x_{ij}} \le 1 \ \ \ \text{and} \ \ \ \forall j \in G: \ \sum\nolimits_{i\in A} {x_{ij}} \le 1 . \]
This simple observation eliminates a large number of variables making the LPs more easily solvable.

In addition, given $f(\vv)=\sum_{i\in A} \log(v_i-c_i)$, extracting a cut  at iteration $t$ of CP  requires $v^{(t)}_i - c_i > 0$ for each agent $i$ (note that cut coefficients are $1/(v^{(t)}_i - c_i)$). However, since CP is an outer-approximation algorithm, it is possible to get $v^{(t)}_i - c_i \le 0$ for some agent $i$, making the solution $\vv^{(t)}$ unbounded, preventing cut extraction from this solution. 
To resolve this issue, we add a bound constraint $v_i\ge \underline{v}_i$, where $\underline{v}_i > c_i$ is an optimality lower bound obtained as described in Section~\ref{sec.CPs-bounding}.

\paragraph{Cut generation.}
Cutting-plane algorithms are notorious for ``zig-zagging'' from one point in the feasible region to another. To stabilize the algorithm,
akin to \cite{ben2007acceleration}, we employ a trust region strategy, where,
instead of cutting off the current solution $\vv^{(t)}$, we cut off an intermediate point $\tilde{\vv}=\tilde{\alpha} \vv^{(t)}+(1-\tilde{\alpha})\vv^*$, where $\vv^*$ is the current incumbent solution. We select $\tilde{\alpha}\in(0,1]$ such that (i) $f(\tilde{\vv})$ is potentially a better lower bound than $f(\vv^{(t)})$, and (ii)  the cut associated with $\tilde{\vv}$ cuts off $(\vv^{(t)},\eta^{(t)})$, that is $\eta^{(t)} > f(\tilde{\vv})-n+\sum_{i\in A}\frac{v^{(t)}_i-c_i}{\tilde{v}_i-c_i}$. Thus we compute 
$$\tilde{\alpha}=\arg\max_{\alpha\in [0,1]} f(\alpha \vv^{(t)} + (1-\alpha)\vv^*).$$
If necessary, we keep increasing $\tilde{\alpha}$ until the cut produced using $(\tilde{\vv},\tilde{\xx})$ cuts off current solution $(\vv^{(t)},\eta^{(t)})$  or until $\tilde{\alpha}=1$.
Thus, we produce more effective cuts and improve the lower bound quickly.



\paragraph{Scaling of $\eta$.} An LP solver using floating point arithmetic might not handle unbalanced cuts properly.
For a given solution $\hat{\vv}$, coefficient of $\eta$ in a cut of the form \eqref{eq.silp} is 1, while the coefficients of the $\vv$-variables are $(\frac{1}{\hat{v}_i-c_i})_{i\in A}$, which can be much larger than 1 depending on the value of $\hat{\vv}$. For instance, when entries of the utility matrices are at most 1 and $c_i>0$, then $\hat{v}_i-c_i<1$, and it is possible that $\frac{1}{\hat{v}_i-c_i} \gg 1$ for some agents, making the cut coefficients unbalanced. 

To balance the cuts, we first scale all utilities so that they are at most 1. Note that the Nash solution is immune to this transformation (Theorem~\ref{thm.nash}). Then, we replace $\eta$ with $\eta=\theta\gamma$, where $\theta>0$ is a fixed scalar and $\gamma$ acts as the new variable in place of $\eta$. With this change of variable, coefficient of $\sigma$ in the cuts becomes $\|(\theta,\nabla f(\hat{\vv}))\|_2$. In our implementation, we choose $\theta$ as the average coefficient of the $\vv$-variables in the first cut produced, that is $\theta=\frac{1}{n}\sum_{i\in A}\frac{1}{\hat{v}_i-c_i}$. While we may dynamically change $\theta$, we use the same initial $\theta$ for stabilizing all subsequent cuts.

\paragraph{Reoptimization.}
At each iteration of Algorithm~\ref{pseudo-code-cutting-plane}, we add a single constraint of the form \eqref{eq:cp-cut} to the current LP approximation of {\em 1LAD}. Instead of solving these LPs from scratch, we reuse the information obtained in the previous iteration (e.g. the basis) using the Dual Simplex algorithm.

\subsubsection{Extension to other models.}
Given the same functional form in the objective functions and the linear constraints in {\em 2LF}, {\em 1SAD}, {\em 1NAD}, and {\em 1LF2G}, Algorithm~\ref{pseudo-code-cutting-plane} and the acceleration techniques previously detailed for {\em 1LAD} extend to these models easily by replacing the objective function and the constraints with the appropriate function and constraints, respectively. For instance, for {\em 2LF} we have the same set of constraints but the cutting planes take the form of 
$$\eta\le \sum\nolimits_{i\in A}\log(\hat{v}_i)+\sum\nolimits_{j\in G}\log(\hat{v}_j)-2n+\sum\nolimits_{i\in A}\frac{v_i}{\hat{v}_i}+\sum\nolimits_{j\in G}\frac{v_j}{\hat{v}_j}.$$
Note that, in {\em 2LF}, we may eliminate the $x_{ij}$ variables when both $u_{ij}$ and $w_{ij}$ are zero. Similarly, we may eliminate $x_{ijl}$ for triple $(i,j,l)$ when $u_{ijl}=0$ in {\em 1LF2G}. For {\em 1SAD} (respectively {\em 1NAD}) we can eliminate $x_{ij}$ when $u_{ijk}=0$ (respectively $a_{ijk}=0$) for all pieces of the utility functions.

\subsection{Initial Perfect Matching}\label{sec:initial_matching}
A high quality initial solution is key to fast convergence of both CP and FW algorithms. Note that $\xx_i=\sum_{j\in G}x_{ij}\ee_j$, where $\ee_j$ is the $j^{\text{th}}$ basic vector in $\R^n$. Concavity of $\log(u_i(\xx)-c_i)$ together with $\sum_{j\in G}x_{ij}=1$ imply that for any $\xx\in \calX$ 
\begin{align}
    \log(u_i(\xx)-c_i) \ge \sum\nolimits_{j\in G} x_{ij}\log(u_i(\ee_j)-c_i) \qquad \forall i\in A,
\end{align}
with equality when $x_{ij}\in \{0,1\}$. Therefore, the linear program 
$$\max\limits_{\xx\in \calX}\quad\sum\nolimits_{i\in A}\sum\nolimits_{j\in G} x_{ij}\log(u_i(\ee_j)-c_i)$$
effectively produces a lower bound and a near optimal solution to $\max\limits_{\xx\in \calX}\sum\nolimits_{i\in A}\log(u_i(\xx)-c_i)$ by pushing the objective value from below.
To avoid unboundedness, we define $\log_M(\alpha)=\log(\alpha)$ for $\alpha>0$ and $\log_M(\alpha)=-M$ for $\alpha\le 0$, in which $M$ is a sufficiently large constant. With this definition, we solve the following maximum-weight matching problem to produce a near optimal solution:
\begin{align}
	\max_{\xx\in\calX} \quad \sum\nolimits_{i\in A}\sum\nolimits_{j\in G}x_{ij}\log_M(u_i(\ee_j)-c_i).\label{op:initial}
\end{align}
%
%

Similarly, for {\em 2LF}, we use the surrogate linear function $\sum_{i\in A}\sum_{j\in G}x_{ij}(\log_M(u_{ij})+\log_M(w_{ij}))$, since $w_j(\xx)$ is also a convex combination of $w_{1j}, w_{2j},\dots,w_{nj}$ with weights $x_{1j},x_{2j},\dots,x_{nj}$. 


We end this section by highlighting the quality and importance of the  solution obtained by \eqref{op:initial}. 
\begin{proposition}\label{prop:initial_bound}
    Provided the matching market model admits an integral optimal solution,  then (i) the lowerbound provided by \eqref{op:initial} is tight, and (ii) both FW and CP algorithms terminate after one iteration when initialized by the solution produced by \eqref{op:initial}.
\end{proposition}


\section{Computational Results}\label{sec.Exp}
We conducted computational experiments on instances of various difficulty levels to assess the scalability of the proposed algorithms for each matching market model.
All experiments were conducted on a Dell desktop equipped with Intel(R) Xeon(R) CPU E5-2680 v3 at 2.50GHz with 8 Cores and 32 GB of memory. 
We coded our algorithms in \texttt{C\#} integrated with th \texttt{ILOG Concert} library for solving the LPs in Algorithm \ref{pseudo-code-cutting-plane} using the Dual Simplex method implemented in \texttt{CPLEX 12.10} solver (\texttt{RootAlgorithm} parameter set to \texttt{Cplex.Algorithm.Dual}). 
We implemented the Auction algorithm \citep{bertsekas1988auction} for solving the perfect matching subproblems in FW. While this algorithm runs in $\calO(n^3)$ in the worst case; we found it to be very efficient in solving extremely large matching problems. 
We terminated the FW and CP algorithms upon reaching either an optimality gap of 0.01\%, running time of 3600 seconds, or 10000 iterations.

\subsection{Computational Results for {\em 1LAD}, {\em 1LF} and {\em 2LF}}
We start by presenting the results for matching market models with linear utility functions.
We performed computational experiments on {\em 1LAD}, {\em 1LF} and {\em 2LF} by producing random utility matrices $u$ (and $w$ in {\em 2LF}) as follows. 
We control the density (i.e. number of positive entries) of the utility matrices using parameter $\rho\in\{\frac{1}{20},\frac{1}{3},\frac{2}{3}\}$. We considered two types of utility matrices:  (\textit{a}) \textbf{binary}, where $u_{ij}$ ($w_{ij})$ was set to 1 with probability $\rho$, and  (\textit{b}) \textbf{nonbinary}, where $u_{ij}$ ($w_{ij})$ was set to a positive value with probability $\rho$ uniformly drawn from $\{1,2,\dots,20\}$.
For {\em 1LAD}, we chose $c_i$ uniformly from the set $\{\frac{\bar{u}}{3},\frac{\bar{u}}{4},0\}$, where $\bar{u}=\frac{1}{4}\max_{ij}\{u_{ij}\}$ to ensure feasibility. Tables~\ref{tab:results-1lad}, \ref{tab:results-1lf} and \ref{tab:results-2lf} provide detailed results for FW and CP in terms of time to produce the initial solution as described in Section~\ref{sec:initial_matching}, optimality gap of the initial solution, and total computation times. 

Average computation times of FW and CP for solving instances of {\em 1LAD}, {\em 1LF} and {\em 2LF} are illustrated in Figure~\ref{fig:res_linear}. As observed, FW is able to solve all instances, even for {\em 2LF} and instances with extremely large number of agents/goods, in less than 4 minutes, thanks to its capacity for exploiting the combinatorial structures of these problems. While CP is also competitive, particularly for smaller instances, its performance heavily relies on efficiency of the LP solver both in terms of running time and memory. This is particularly evident for larger instances where computation time grows an order of magnitude faster compared to FW and the solver runs out of memory for $n\ge 5000$. 
It is also interesting to note that FW exhibits similar performance in solving instances with both binary and nonbinary utility matrices, whereas binary utility matrices tend to be harder than nonbinary utility matrices for CP. The optimality gaps reported for the initial solutions in Tables~\ref{tab:results-1lad}, \ref{tab:results-1lf} and \ref{tab:results-2lf} also show that the initial solutions that we produce are near optimal which highlight the importance and effectiveness of the procedure introduced in Section~\ref{sec:initial_matching}.

As explained in Section~\ref{sec:fw}, FW also has the advantage of producing rational solutions that are sparse convex combinations of some integral perfect matchings. 
Figure~\ref{fig:example} in Appendix~\ref{app:results} illustrates an example where the solutions produced by CP and FW are numerically optimal (resp. -1.0939295 and -1.0939293), but the solution produced by FW is the true rational optimal solution.

\begin{figure}[h]
    \centering
    \includegraphics[width=0.9\textwidth]{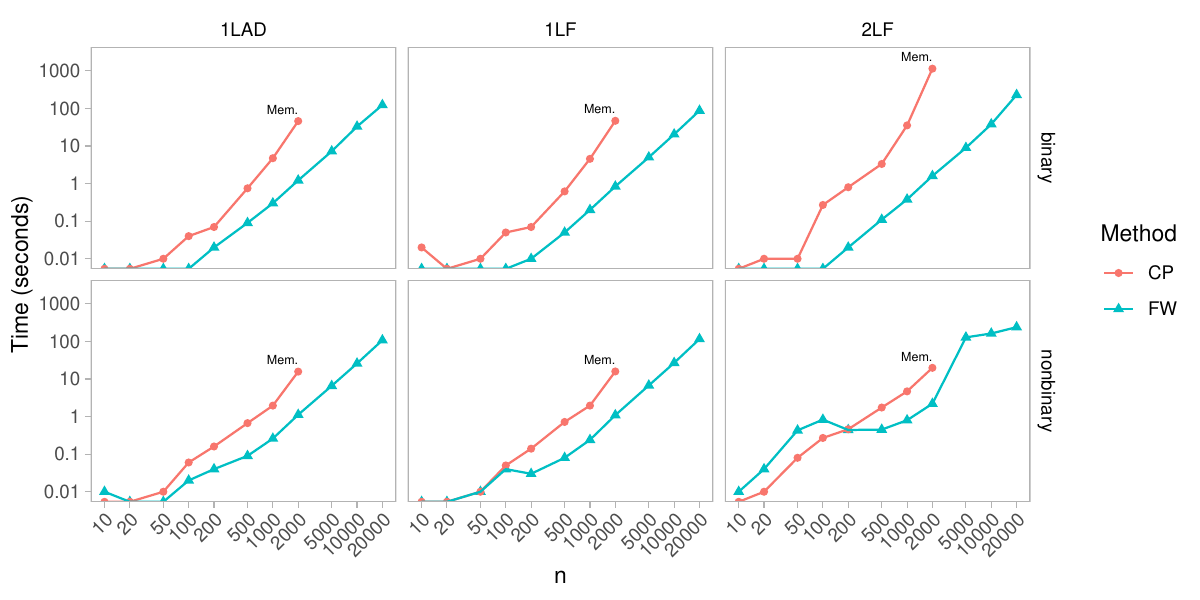}
    \caption{Computation times for {\em 1LAD}, {\em 1LF} and {\em 2LF}. Both axes are log-scaled.}
    \label{fig:res_linear}
\end{figure}

\subsection{Computational Results for {\em 1SAD} and {\em 1NAD}}
We generated random instance for {\em 1SAD} by constructing piecewise linear concave utility functions each with $K\in\{5,10,20,50\}$ segments of equal length. To generate non-decreasing utility rates $u_{ij1}>u_{ij2}>\dots>u_{ijK}$ for each $(i,j)$, we first generated $K$ values $\nu_{k}$ randomly drawn from $\{1,\dots,20\}$, and then set $u_{ijk}=\sum_{l\in \range{K}}\nu_{l}$. 
We then scaled the $u_{ijk}$ values such that the area under the utility function becomes $\frac{1}{2}\tilde{u}$ for some $\tilde{u}$ uniformly drawn from $\{1,\dots,20\}$.

For {\em 1NAD}, we considered $K\in\{5,10,20,50\}$ hyperplanes of the form $\sum_{j\in G}a_{ijk}x_{ij}+b_{ik}$ for each $i\in A$, and generated the coefficients $a_{ijk}$ by multiplying $\frac{2}{3}$ with a value uniformly drawn from the set $\{0,1,\dots,20\}$, and generated the intercept $b_{ik}$ by multiplying $\frac{1}{3}$ with a value uniformly drawn from the set $\{0,1,\dots,20\}$. If $b_{ik}>0$ for all $k$, then we randomly set one of them to 0.

Figure~\ref{fig:res_sad_nad} illustrates the computation times for {\em 1SAD} and {\em 1NAD} across different choices of $n$ and $K$ using the CP and FW algorithms. 
Interestingly, instances of {\em 1NAD} are in general more challenging than {\em 1SAD}, for both CP and FW. For CP, this is because the constraints in the master problems have a sparser structure in {\em 1SAD}. For FW, this is because the Moreau envelopes can be computed analytically for {\em 1SAD} as opposed to solving quadratic programs for {\em 1NAD}. 

CP turns out to be quite efficient for small and moderately sized instances, but FW starts to outperform CP as $n$ and $K$ become larger. As in the previous experiments, this is because the LP solver runs out of memory or hits the time limit when the master problem contains too many variables/constraints. Conversely, FW is less impacted by $K$, thanks to independency of the matching problems to the number of segments in the utility functions.


\begin{figure}[h]
    \centering
    \includegraphics[width=0.9\textwidth]{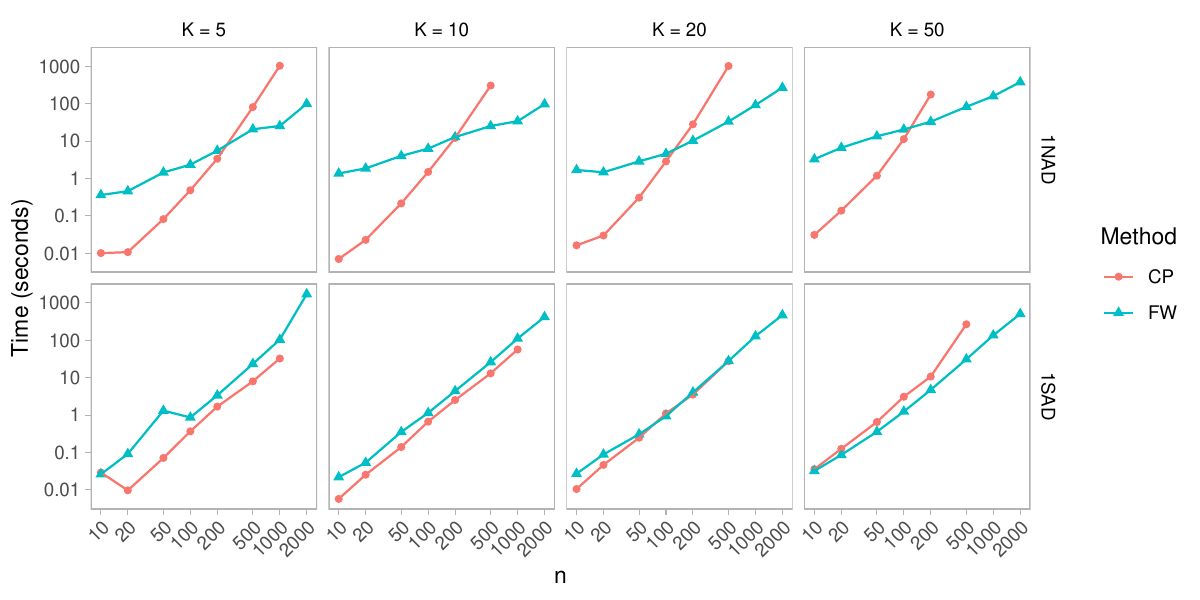}
    \caption{Computation times for {\em 1SAD} and {\em 1NAD}. Both axes are log-scaled.}
    \label{fig:res_sad_nad}
    \vspace{0.25cm}
    \centering
    \includegraphics[width=0.7\textwidth]{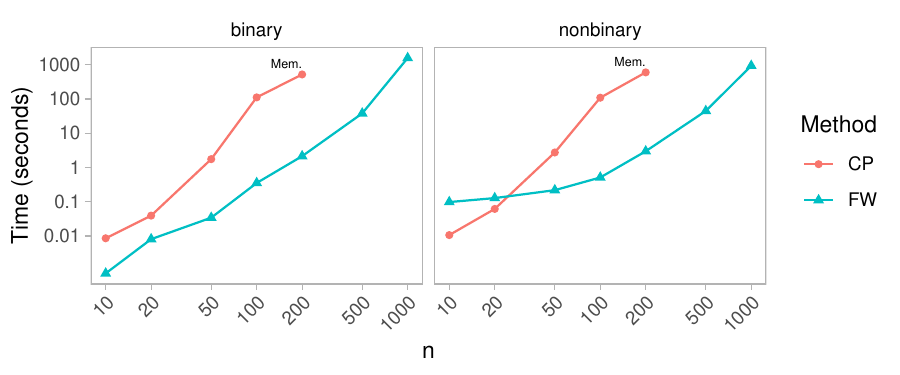}
    \caption{Computation times for {\em 1LF2G}. Both axes are log-scaled.}
    \label{fig:res_1lf2g}
\end{figure}

\subsection{Computational Results for {\em 1LF2G}}
Figure~\ref{fig:res_1lf2g} illustrates the performance of generalized FW (Algorithm~\ref{pseudo-code-entropic-frank-wolfe}) and CP (Algorithm~\ref{pseudo-code-cutting-plane}) on solving instances of {\em 1LF2G}. The instances are generated as in {\em 1LF} considering binary and nonbinary utilities. Given the large number of variables ($n^3$) involved in this model, CP reaches its limit as soon as the number of agents surpasses $n=200$. However, our specialized multimarginal matching routine allows us to solve instances with as many as $n=1000$ agents in less than 30 minutes using our generalized FW algorithm, which further highlights efficacy of our implementation. 

\section{Conclusions}\label{sec:conclusion}
This paper is motivated by three factors: first, the revival, in the last few years, of the area of matching-based market design; second, the recent realization that the most prominent solution that uses cardinal utilities, namely the Hylland-Zeckhauser (HZ) mechanism \citep{hylland} for a one-sided matching markets, is intractable; and third, the extreme paucity of matching market models that use cardinal utilities.

As a consequence of the revival, e.g., see the upcoming book \cite{Echenique2023online}, it is safe to assume that new markets will be introduced in the economy and will call for new models and efficient mechanisms with good properties. The papers \cite{VY-HZ} and \cite{HZ-hardness} show that the problem of computing even an approximate HZ equilibrium is PPAD-complete. In practice,  computing an HZ equilibrium for linear utilities for even $n = 10$ is prohibitive. The paucity of  models stands in sharp contrast with general equilibrium theory, which has defined and extensively studied several fundamental market models to address a number of specialized and realistic situations. 

Our paper addresses these issues by proposing Nash-bargaining-based matching market models. The Nash bargaining solution can be computed efficiently since it is captured by a convex program. Furthermore, our approach yields a rich collection of models: for one-sided as well as two-sided markets, for Fisher as well as Arrow-Debreu settings, and for a wide range of utility functions, all the way from linear to Leontief. 

We also give very fast implementations for these models using Frank-Wolfe and Cutting Plane algorithms. These help solve large instances with several thousand agents and goods in a matter of minutes on a PC, even for a one-sided matching market under piecewise-linear concave utility functions and a two-sided matching market under linear utility functions. A number of new ideas were needed, beyond the standard methods, to obtain these implementations. In particular, we present several lower bounding schemes, which not only help improve the convergence of our solution methods but also highlight fairness properties of the Nash-bargaining-based models.

\bibliographystyle{informs2014}
\bibliography{refs}

\ACKNOWLEDGMENT{Supported in part by NSF grant CCF-1815901. The second author would like to thank Richard Zeckhauser for a very enlightening discussion on his ``wish list'' of models for matching markets. Several of the models, studied in this paper, have their origins in that discussion. }

\ECSwitch




%


%
%

\begin{APPENDICES}

\section{Proofs of Statements}

\subsection{Proof of Lemma \ref{lemma:utility_bound-feas}}

\begin{repeattheorem}[Lemma \ref{lemma:utility_bound-feas}.]
   Let $\xx^*$ be an optimal solution to $\max\limits_{\xx \in \calX}\; \sum_{i\in A}\log(u_i(\xx))$, 
    where $u_i(\xx)$ is a concave utility function with $u_i(\vecz)=0$.
    For any $\bar{\xx}\in \R^{n^2}_+$ (not necessarily a perfect matching), we have
    \begin{align}
        u_i(\xx^*)\ge \frac{1}{\sum_{j}\bar{x}_{ij}+n\max_{j}\{\bar{x}_{ij}\}} u_i(\bar{\xx}). \label{app-eq:concave-util-bound-feas}
    \end{align}
\end{repeattheorem}
\prf{
Let $\alpha_i\ge 0$ and $\beta_j\ge 0$ be the Lagrangian multipliers of constraints $\sum_{j} x_{ij}\le 1$ and $\sum_{i} x_{ij}\le 1$, respectively. Note that these constraints are tight at optimality since the objective function maximizes positive concave utilities. Furthermore, stationary condition implies that there exists $g_{ij}$, a subderivative of $u_i$ at $\xx^*$, such that
\begin{align}
    \frac{g_{ij}}{u_i(\xx^*)} \le \alpha_i +\beta_j \quad \forall i,j  \label{app-eq:concave-util-bound-stationary}
\end{align}
with equality when $x^*_{ij}>0$. Concavity of $u_i$ implies that 
\begin{align}
    u_i(\bar{\xx})\le &u_i(\xx^*)  + \sum_{j}g_{ij}\bar{x}_{ij} - \sum_{j} g_{ij} x^*_{ij}.\label{app-eq:concave-util-bound-feas-concavity}
\end{align}
In particular, given that $u_i(\vecz)=0$, \eqref{app-eq:concave-util-bound-feas-concavity} implies that $\sum_{j}g_{ij}x^*_{ij}\le u_i(\xx^*)$.
Multiplying both sides of \eqref{app-eq:concave-util-bound-stationary} by $x^*_{ij}$ and summing over $j$ we obtain
\begin{align}
    \frac{\sum_{j}g_{ij}x^*_{ij}}{u_i(\xx^*)} = \sum_{j}x^*_{ij}(\alpha_i +\beta_j) = \alpha_i + \sum_{j}x^*_{ij}\beta_j \le 1  \label{app-eq:concave-util-bound-4}
\end{align}
where we have used $\frac{\sum_{j}g_{ij}x^*_{ij}}{u_i(\xx^*)}\le 1$ and $\sum_{j}x^*_{ij}=1$. In particular, \eqref{app-eq:concave-util-bound-4} implies that $\alpha_i\le 1$ since $\beta_j\ge 0$. Additionally, summing \eqref{app-eq:concave-util-bound-4} over $i$ we get $\sum_{i}\alpha_i + \sum_{j}\beta_j \le n$ since $\sum_{i}x^*_{ij}=1$.
From \eqref{app-eq:concave-util-bound-4} we obtain
\begin{align}
    \sum_{j}g_{ij}x^*_{ij} = u_i(\xx^*)(\alpha_i + \sum_{j}x^*_{ij}\beta_j). \label{app-eq:concave-util-bound-5}
\end{align}
On the other hand, multiplying both sides of \eqref{app-eq:concave-util-bound-stationary} by $u_i(\xx^*)\bar{x}_{ij}$ and summing over $j$ we obtain
\begin{align}
    \sum_j g_{ij} \bar{x}_{ij} \le u_i(\xx^*)(\alpha_i\sum_{j}\bar{x}_{ij} +\sum_j\bar{x}_{ij}\beta_j).  \label{app-eq:concave-util-bound-stationary-x}
\end{align}
Substituting \eqref{app-eq:concave-util-bound-5} and \eqref{app-eq:concave-util-bound-stationary-x} into \eqref{app-eq:concave-util-bound-feas-concavity} we obtain
\begin{align*}
   u_i(\bar{\xx})\le u_i(\xx^*)\left(1+\alpha_i \sum_{j}\bar{x}_{ij} +\sum_j\bar{x}_{ij}\beta_j - \alpha_i-\sum_{j}x^*_{ij}\beta_j\right)\le \left(\sum_{j}\bar{x}_{ij}+n\max_{j}\{\bar{x}_{ij}\}\right) u_i(\xx^*),
\end{align*}
where we have used $1+(\sum_{j}\bar{x}_{ij}-1)\alpha_i\le \sum_{j}\bar{x}_{ij}$ since $\alpha_i\le 1$, and $\sum_j\bar{x}_{ij}\beta_j\le \max_{j}\{\bar{x}_{ij}\}\sum_{j}\beta_j\le \max_{j}\{\bar{x}_{ij}\}n$.
}

\subsection{Proof of Proposition~\ref{prop:lowerbound-1lf-kp}}

\begin{repeattheorem}[Proposition~\ref{prop:lowerbound-1lf-kp}.]
Let $G_m$ denote the indices of the top $m$ utilities in $\{u_{ij}\}_{j\in G}$. Then
\begin{align*}
    u_i(\xx^*)\ge \max_{m\in \range{n}} \left\{\frac{1}{n+m}\sum_{j\in G_m}u_{ij}\right\}.
\end{align*}
\end{repeattheorem}
\prf{The proof follows by evaluating \eqref{eq:concave-util-bound-feas} at $\yy^{(m)}$ for each $m\in \range{n}$, where  $\yy^{(m)}$ is constructed such that $y^{(m)}_j=\frac{1}{m}$ for $j\in G_m$ and $y^{(m)}_j=0$ for $j\ne G_m$.
}

\subsection{Proof of Lemma \ref{lemma:utility_bound-feas-c}}

\begin{repeattheorem}[Lemma \ref{lemma:utility_bound-feas-c}.]
    Let $\xx^*$ be an optimal solution to $\max_{\xx \in \calX}\; \sum_{i\in A}\log(u_i(\xx)-c_i)$.
    For any $\hat{\xx}\in \calX$ we have
    \begin{align*}
        u_i(\xx^*)-c_i \ge \frac{1}{1+\frac{n}{1-\hat{\sigma}\rho}\max_{j}\{\hat{x}_{ij}\}} (u_i(\hat{\xx})-c_i). 
    \end{align*}
    where $\hat{\sigma}= \max\limits_{i\in A}\left\{\frac{c_i}{u_i(\hat{\xx})}\right\}$ is a disagreement ratio for $\hat{\xx}$ and  $\rho = \max\limits_{i\in A, j\in G}\left\{\frac{g^{\max}_{ij}}{g^{\min}_{ij}}\right\}$ is a curvature ratio for the utility functions, in which $g^{\max}_{ij}$ and $g^{\min}_{ij}$ are the highest and lowest utility rates for agent $i$ from good $j$ in $u_{i}$, respectively.
\end{repeattheorem}
\prf{As before, stationary condition implies
\begin{align}
    \frac{g_{ij}}{u_i(\xx^*)-c_i} \le \alpha_i +\beta_j \quad \forall i,j  \label{app-eq:concave-util-bound-feas-c-stationary}
\end{align}
with equality when $x^*_{ij}>0$. 
Multiplying both sides of \eqref{app-eq:concave-util-bound-feas-c-stationary} respectively by $(u_i(\xx^*)-c_i)x^*_{ij}$ and $(u_i(\xx^*)-c_i)\bar{x}_{ij}$ and summing over $j$ we obtain
\begin{align}
     \sum_j g_{ij} x^*_{ij} = (u_i(\xx^*)-c_i)(\alpha_i +\sum_j x^*_{ij}\beta_j),  \label{app-eq:concave-util-bound-feas-c-x1}\\
    \sum_j g_{ij} \bar{x}_{ij} \le (u_i(\xx^*)-c_i)(\alpha_i +\sum_j\bar{x}_{ij}\beta_j),  \label{app-eq:concave-util-bound-feas-c-x2}
\end{align}
where we have used $\sum_{j}x^*_{ij}=\sum_{j}\bar{x}_{ij}=1$.
Concavity of $u_i$ implies that 
\begin{align}
    u_i(\bar{\xx})-c_i\le &u_i(\xx^*)-c_i  + \sum_{j}g_{ij}\bar{x}_{ij} - \sum_{j} g_{ij} x^*_{ij}.\label{app-eq:concave-util-bound-feas-c-concavityx1x2}
\end{align}
Substituting \eqref{app-eq:concave-util-bound-feas-c-x1} and \eqref{app-eq:concave-util-bound-feas-c-x2} into \eqref{app-eq:concave-util-bound-feas-c-concavityx1x2} we obtain
\begin{align}
   u_i(\bar{\xx})-c_i\le & (u_i(\xx^*)-c_i)\left(1+\alpha_i  +\sum_j\bar{x}_{ij}\beta_j - \alpha_i-\sum_{j}x^*_{ij}\beta_j\right)\nonumber \\\le & \left(1+\max_{j}\{\bar{x}_{ij}\}\sum_{j}\beta_j\right) (u_i(\xx^*)-c_i) \label{app-eq:concave-util-bound-feas-c-concavity-bound}
\end{align}
It remains to bound $\sum_{j}\beta_j$, for which we first bound $\sum_{i}\alpha_i+\sum_{j}\beta_j$ using $\hat{\xx}$. 
From $\sum_{j}g_{ij}x^*_{ij}\le u_i(\xx^*)$ we deduce $\frac{\sum_{j}g_{ij}x^*_{ij}}{u_i(\xx^*)-c_i}\le 1+\frac{c_i}{u_i(\xx^*)-c_i}$, hence we obtain
\begin{align*}
    \sum_{i}\alpha_i + \sum_{j}\beta_j= &\sum_{i}(\alpha_i + \sum_{j}x^*_{ij}\beta_j)=\sum_{i}\frac{\sum_{j}g_{ij}x^*_{ij}}{u_i(\xx^*)-c_i}\le n+\sum_{i\in A}\frac{c_i}{u_i(\xx^*)-c_i}.
\end{align*}
To bound $\frac{c_i}{u_i(\xx^*)-c_i}$ we use the curvature parameter $\rho$ and the disagreement bound $\hat{\sigma}$. First, from concavity of $u_i$ together with $u_i(\vecz)=0$ we obtain
\begin{align*}
    u_i(\hat{\xx})\le u_i(\vecz) + \sum_{j\in G} \partial_{ij} u_i(\vecz) \hat{x}_{ij}\le \sum_{j\in G}\hat{x}_{ij}g^{\max}_{ij}.
\end{align*}
Now, from \eqref{app-eq:concave-util-bound-feas-c-stationary} we obtain in order
\begin{align}
    \frac{c_i}{u_i(\xx^*)-c_i} \le & c_i\min_{j\in G}\left\{\frac{\alpha_i+\beta_j}{g_{ij}}\right\} \le \hat{\sigma}u_i(\hat{\xx})\min_{j\in G}\left\{\frac{\alpha_i+\beta_j}{g^{\min}_{ij}}\right\} \le \hat{\sigma}\sum_{j\in G}\hat{x}_{ij}g^{\max}_{ij}\min_{l\in G}\left\{\frac{\alpha_i+\beta_l}{g^{\min}_{il}}\right\} \nonumber \\ \le & \hat{\sigma}\sum_{j\in G}\hat{x}_{ij}\frac{g^{\max}_{ij}}{g^{\min}_{ij}}(\alpha_i+\beta_j)\le \hat{\sigma}\rho\sum_{j\in G}\hat{x}_{ij}(\alpha_i+\beta_j)=\hat{\sigma}\rho(\alpha_i+\sum_{j\in G}\hat{x}_{ij}\beta_j)\label{app-eq:concave-util-bound-feas-c-cu}
\end{align}
where we have used $c_i\le \hat{\sigma}u_{i}(\hat{\xx})$, $g_{ij}\ge g^{\min}_{ij}$, $\frac{g^{\max}_{ij}}{g^{\min}_{ij}}\le \rho$ and $\sum_{j}\hat{x}_{ij}=1$.
Finally, summing \eqref{app-eq:concave-util-bound-feas-c-cu} over $i$ and noting that $\sum_{i} \hat{x}_{ij}=1$ we deduce
\begin{align*}
    \sum_i \alpha_i+\sum_{j}\beta_j\le & n+\sum_{i\in A}\frac{c_i}{u_i(\xx^*)-c_i}\le n+\hat{\sigma}\rho(\sum_{i}\alpha_i+\sum_{j\in G}\beta_j) \quad \Rightarrow \nonumber\\
    \sum_i \alpha_i+\sum_{j}\beta_j\le & \frac{n}{1-\hat{\sigma}\rho} \quad \Rightarrow \quad \sum_{j}\beta_j\le \frac{n}{1-\hat{\sigma}\rho}.
\end{align*}
Plugging $\sum_{j}\beta_j\le \frac{n}{1-\hat{\sigma}\rho}$ into \eqref{app-eq:concave-util-bound-feas-c-concavity-bound} shows the result.
}

\subsection{Proof of Lemma \ref{lemma:utility_bound-multiple}}

\begin{repeattheorem}[Lemma \ref{lemma:utility_bound-multiple}.]
    Let $\xx^*$ be an optimal solution to a one-sided matching market over $m$ goods with $u_i(\xx)$ a concave utility function such that $u_i(\vecz)=0$. Then
    \begin{align*}
        u_i(\xx^*)\ge \frac{1}{2n^{m}}\max_{\theta >0}\left\{\frac{u_i(\theta\ee)}{\theta}\right\}, 
    \end{align*}
    where $\ee$ is the vector of all ones in $\R^{n^{m+1}}$. (Note: $\theta\ee$ need not be feasible.)
\end{repeattheorem}
\prf{For ease of exposition, we prove the statement for $m=2$ goods; the proof for $m>2$ follows similarly. 
Concavity of $u_i$ implies
\begin{align}
   u_i(\theta\ee) \le u_i(\xx^*) + \theta\sum_{jl}g_{ijl} - \sum_{jl} g_{ijl} x^*_{ijl},\label{app-eq:concave-util-bound-multiple-1}
\end{align}
where $g_{ijl}$ is a subderivative of $u_i$ at $\xx^*$.
In particular, for $\theta=0$ we obtain $u_i(\xx^*)\ge \sum_{jl} g_{ijl} x^*_{ijl}$ since $u_i(\vecz)=0$. 
Stationary condition implies
\begin{align}
    \frac{g_{ijl}}{u_i(\xx^*)} \le \alpha_i +\beta_j + \sigma_l \quad \forall i,j,l  \label{app-eq:concave-util-bound-multiple-2}
\end{align}
with equality when $x^*_{ijl}>0$, where $\alpha_i\ge 0$, $\beta_j\ge 0$ and $\sigma_l\ge 0$ are the Lagrangian multipliers of constraints $\sum_{jl} x_{ijl}\le 1$, $\sum_{il} x_{ijl}\le 1$, and $\sum_{ij} x_{ijl}\le 1$ respectively. Note that these constraints are tight at optimality since the objective function maximizes positive concave utilities. 
As before, from \eqref{app-eq:concave-util-bound-multiple-2} we obtain
\begin{align}
    u_i(\xx^*) \left(n^2\alpha_i + n\sum_{j}\beta_j+n\sum_{l}\sigma_l\right) \ge &\sum_{jl}g_{ijl} & \forall i \label{app-eq:concave-util-bound-multiple-3}\\
    u_i(\xx^*) \left(\alpha_i + \sum_{jl}x^*_{ijl}(\beta_j+\sigma_l)\right) = &\sum_{jl}g_{ijl}x^*_{ijl} & \forall i \label{app-eq:concave-util-bound-multiple-4}\\
    \alpha_i + \sum_{jl}x^*_{ijl}(\beta_j+\sigma_l) \le & 1 & \forall i \label{app-eq:concave-util-bound-multiple-5}\\
    \sum_{i}\alpha_i + \sum_{j}\beta_j+\sum_{l}\sigma_l \le & n\label{app-eq:concave-util-bound-multiple-6}
\end{align}
Plugging \eqref{app-eq:concave-util-bound-multiple-3} and \eqref{app-eq:concave-util-bound-multiple-4} into \eqref{app-eq:concave-util-bound-multiple-1} we obtain
\begin{align*}
    u_i(\theta\ee)\le &u_i(\xx^*)\left(1+(\theta n^2-1)\alpha_i  +n\theta(\sum_j\beta_j+\sum_{l}\sigma_l) - \sum_{jl}x^*_{ijl}(\beta_j+\sigma_l)\right)\le 2n^2\theta u_i(\xx^*),
\end{align*}
where we have used $1+(\theta n^2-1)\alpha_i\le n^2\theta$ since $\alpha_i\le 1$ by \eqref{app-eq:concave-util-bound-multiple-5}, and $\sum_{j}\beta_j+\sum_{l}\sigma_l\le n$ by \eqref{app-eq:concave-util-bound-multiple-6}.
}

\subsection{Proof of Lemma \ref{lemma:bound-2lf}}

\begin{repeattheorem}[Lemma \ref{lemma:bound-2lf}.]
Let $\xx^*\in \calX$ be an optimal solution to $ \max\limits_{\xx \in \calX}\; \sum\limits_{i\in A}\log(u_i(\xx))+\sum\limits_{j\in J}\log(w_j(\xx))$, and define $\bar{w}_{ij}=\frac{w_{ij}}{\max\limits_{i'\in A}\{w_{i'j}\}}$ and $\bbar{w}_{i}=\min\{n, \max\limits_{j\in J}\{\frac{w_{ij}}{\min\limits_{i'\in A}\{w_{i'j}\}}\}\}$. 
 Then, for any $\bar{\xx}\in\R^{n^2}_+$ and $i\in A$:
\begin{align}
    u_i(\xx^*)\ge &\frac{u_i(\bar{\xx})}{\sum\limits_{j\in J}\bar{x}_{ij}(2n-\bar{w}_{ij})}, \label{app-eq:bound-2lf-u-1}\\
    u_i(\xx^*)\ge &\frac{u_i(\bar{\xx})}{\sum\limits_{j\in J}\bar{x}_{ij}(1+\bbar{w}_i-\bar{w}_{ij})+2n\max\limits_{j\in J}\{\bar{x}_{ij}\}}. \label{app-eq:bound-2lf-u-2}
\end{align}
\end{repeattheorem}
\prf{As before, stationary condition implies
\begin{align}
    \frac{u_{ij}}{u_i(\xx^*)} + \frac{w_{ij}}{w_j(\xx^*)} \le \alpha_i +\beta_j \qquad \forall i,j  \label{app-eq:bound-2lf-stationary}
\end{align}
with equality when $x^*_{ij}>0$. 
Multiplying both sides of \eqref{app-eq:bound-2lf-stationary} by $x^*_{ij}$ and summing over $j$ we obtain
\begin{align}
    \alpha_i + \sum_{j}\beta_jx^*_{ij}=\frac{\sum_{j}u_{ij}x^*_{ij}}{u_i(\xx^*)} + \sum_{j}\frac{w_{ij}x^*_{ij}}{w_j(\xx^*)}=1+\sum_{j}\frac{w_{ij}x^*_{ij}}{w_j(\xx^*)},\label{app-eq:bound-2lf-stationary-1}
\end{align}
where we have used $u_i(\xx^*)=\sum_{j}u_{ij}x^*_{ij}$. In particular, given that $w_j(\xx^*)=\sum_{i}w_{ij}x^*_{ij}$, we deduce $\sum_{j}\frac{w_{ij}x^*_{ij}}{w_j(\xx^*)}\le n$. On the other hand, $\sum_{j}\frac{w_{ij}x^*_{ij}}{w_j(\xx^*)}\le \sum_{j}\frac{w_{ij}}{\min\limits_{i'\in A}\{w_{i'j}\}}x^*_{ij}\le \max\limits_{j\in J}\{\frac{w_{ij}}{\min\limits_{i'\in A}\{w_{i'j}\}}\}$ since $\sum_{j}x^*_{ij}=1$. Thus we obtain $\alpha_i\le 1+\bbar{w}_i$. Additionally, summing \eqref{app-eq:bound-2lf-stationary-1} over $j$ yields
\begin{align}
     \sum_{i}\alpha_i + \sum_{j}\beta_j= 2n. \label{app-eq:bound-2lf-alpha-beta}
\end{align}
Given $w_j(\xx^*)=\sum_{i\in A}w_{ij}\xx^*_{ij}\le \max_{i\in A}\{w_{ij}\}$ we obtain $\frac{w_{ij}}{w_j(\xx^*)}\ge \bar{w}_{ij}$ and deduce from \eqref{app-eq:bound-2lf-stationary}
\begin{align}
    \frac{u_{ij}}{u_i(\xx^*)} + \bar{w}_{ij} \le \alpha_i +\beta_j \qquad \forall i,j  \label{app-eq:bound-2lf-stationary-2}
\end{align}
Multiplying \eqref{app-eq:bound-2lf-stationary-2} by $\bar{x}_{ij}$, summing over $j$, and noting that $u_i(\bar{\xx})=\sum_{j\in J}\bar{x}_{ij}u_{ij}$ we obtain
\begin{align}
    \frac{u_i(\bar{\xx})}{u_i(\xx^*)}\le \sum_{j\in J}\bar{x}_{ij}(\alpha_i +\beta_j - \bar{w}_{ij})  \label{app-eq:bound-2lf-ux}
\end{align}
Noting that $\alpha_i+\beta_j\le 2n$ by \eqref{app-eq:bound-2lf-alpha-beta}, we obtain \eqref{app-eq:bound-2lf-u-1}. On the other hand, from $\alpha_i\le 1+\bbar{w}_i$ and $\sum_{j}\bar{x}_{ij}\beta_{j}\le \max_{j}\{\bar{x}_{ij}\}\sum_{j}\beta_{j}\le 2n  \max_{j}\{\bar{x}_{ij}\}$ we obtain \eqref{app-eq:bound-2lf-u-2}.
}

\subsection{Proof of Proposition~\ref{prop:lowerbound-2lf-kp}}

\begin{repeattheorem}[Proposition~\ref{prop:lowerbound-2lf-kp}.]
Define $v_{ij}=\frac{u_{ij}}{1+\bbar{w}_i-\bar{w}_{ij}}$ and let $J_m$ denote the indices of the top $m$ values in $\{v_{ij}\}_{j\in J}$. Then
\begin{align*}
    u_i(\xx^*)\ge \max_{m\in \range{n}} \left\{\frac{1}{2n+\sum_{j\in J_m}(1+\bbar{w}_i-\bar{w}_{ij})}\sum_{j\in J_m}u_{ij}\right\}.
\end{align*}
\end{repeattheorem}
\prf{The proof follows by evaluating \eqref{eq:bound-2lf-parametric-kp} at $\theta_m = \frac{1}{2n+\sum_{j\in J_m}(1+\bbar{w}_i-\bar{w}_{ij})}$ for each $m\in \range{n}$, where $\theta_m$ is obtained by setting $y_j=\theta$ for $j\in J_m$ and $y_j=0$ for $j\notin J_m$ and solving $\sum_{j\in J_m} (1+\bbar{w}_i-\bar{w}_{ij})\theta= 1-2n\theta$ for $\theta$.
}

\subsection{Proof of Remark \ref{remark:log_lipschitz}}

\begin{repeattheorem}[Remark \ref{remark:log_lipschitz}.]
    $\log_+(v|v_0)$ defined below has $\frac{1}{v_0^2}$-Lipschitz continuous gradient for $v_0>0$.
    \begin{align*}
    \log_+(v|v_0) = \begin{cases}
        \log(v_0)-1+\frac{v}{v_0} & \text{if } v \le v_0\\
        \log(v) & \text{if } v \ge v_0
    \end{cases}
\end{align*}
\end{repeattheorem}
\prf{Let $f(v)=\log_+(v|v_0)$. We need to show that $|f'(v)-f'(w)| \le \frac{1}{v_0^2}|v-w|$ for all $v$ and $w$. Without loss of generality assume that $v<w$. The statement is trivially correct for $w\le v_0$. 
For $v\le v_0$ and $w>v_0$, we have
$|f'(v)-f'(w)| = |\frac{1}{v_0}-\frac{1}{w}| = \frac{w-v_0}{v_0 w}\le \frac{w-v_0}{v_0^2} \le \frac{1}{v_0^2}|v-w|$.
For $w>v\ge v_0$, we have $|f'(v)-f'(w)| = |\frac{1}{v}-\frac{1}{w}|=\frac{1}{vw}(w-v)\le \frac{1}{v_0^2}|v-w|$.
}

\subsection{Proof of Proposition \ref{prop:smooth-nonseparable}}
\begin{repeattheorem}[Proposition \ref{prop:smooth-nonseparable}.]
    Let $u_i(\xx_i)=\min_{k\in \range{\kappa_i}}\left\{\sum_{j\in G}a_{ijk} x_{ij} + b_{ik}\right\}$ be the non-separable piecewise linear concave utility function of agent $i$ in the 1NAD model, and let $u_{i\gamma}(\xx_i)=-e_{\gamma;-u_i}(\xx_i)$ be its smooth concave envelope. Then
    \begin{align}
        u_{i\gamma}(\xx_i) = & \min_{\zz\in \Delta_{\kappa_i}}\left\{\frac{\gamma}{2}\zz^{\top}Q_i\zz+\xx_i^{\top}A_i \zz\right\}, \label{app-eq:Moreau-env-1nad}\\
        \nabla u_{i\gamma}(\xx) = & A_i\zz^*, \label{app-eq:Moreau-env-1nad-grad}
    \end{align}
    where $\Delta_{\kappa_i}=\{\zz\in \R_+^{\kappa_i}: \sum_{k\in \range{\kappa_i}}z_k=1\}$ is the unit simplex, $A_i$ is an $n\times \kappa_i$ matrix with elements $A_{ijk}=a_{ijk}+b_{ik}$, $Q_i=A_i^{\top}A_i$, and $z^*\in \Delta_{\kappa_i}$ is the optimal solution to \eqref{app-eq:Moreau-env-1nad}.
\end{repeattheorem}
\prf{Given that $\sum\limits_{j\in G}x_{ij}=1$, we have $\sum\limits_{j\in G}a_{ijk} x_{ij} + b_{ik}=\sum\limits_{j\in G}(a_{ijk}+b_{ik}) x_{ij}$; thus $-u_i(\xx_i)$ is
\begin{align*}
    -u_i(\xx_i)=-\min_{k\in \range{\kappa_i}}\left\{\sum_{j\in G}(a_{ijk}+ b_{ik}) x_{ij}\right\}=-\min_{\zz\in \Delta_{\kappa_i}}\left\{\xx_i^{\top}A_i \zz\right\}=\max_{\zz\in \Delta_{\kappa_i}}\left\{-\xx_i^{\top}A_i \zz\right\}.
\end{align*}
Therefore, Moreau envelope of $-u_i(\xx_i)$ is
\begin{align}
    e_{\gamma;-u_i}(\xx_i) &=\min_{\yy\in \R^n}\left\{\frac{1}{2\gamma}\|\xx_i-\yy\|^2+\max_{\zz\in \Delta_{\kappa_i}}\left\{-\yy^{\top}A_i \zz\right\}\right\}\nonumber \\
    &=\min_{\yy\in \R^n}\max_{\zz\in \Delta_{\kappa_i}}\left\{\frac{1}{2\gamma}\|\xx_i-\yy\|^2-\yy^{\top}A_i \zz\right\}=\max_{\zz\in \Delta_{\kappa_i}}\min_{\yy\in \R^n}\left\{\frac{1}{2\gamma}\|\xx_i-\yy\|^2-\yy^{\top}A_i \zz\right\}, \label{app-eq:Moreau-env-1nad-proof-1} 
\end{align}
where we have swapped min-max with max-min because of strong duality. 
The inner quadratic minimization in \eqref{app-eq:Moreau-env-1nad-proof-1} admits optimal solution $\yy=\xx_i+\gamma A_i \zz$. Substituting $\yy$ in \eqref{app-eq:Moreau-env-1nad-proof-1} we obtain
\begin{align*}
    e_{\gamma;-u_i}(\xx_i)=\max_{\zz\in \Delta_{\kappa_i}}\left\{-\frac{\gamma}{2}\zz^{\top}Q_i\zz-\xx_i^{\top}A_i \zz\right\},
\end{align*}
which gives the expression for $u_{i\gamma}(\xx_i)$ after multiplying both sides by $-1$. Finally, for the optimal solution $\zz^*$, given that the optimal value for $\yy$ is $\yy^*=\prx_{\gamma;-u_i}(\xx_i)=\xx_i+\gamma A_i \zz^*$, we obtain $\nabla u_{i\gamma}(\xx_i)=-\nabla e_{\gamma;-u_i}(\xx_i) = -\frac{1}{\gamma}\left(\xx-\prx_{\gamma;-u_i}(\xx_i)\right) = A_i \zz^*$.
}

\subsection{Proof of Proposition \ref{prop:smooth-separable}}
\begin{repeattheorem}[Proposition \ref{prop:smooth-separable}.]
     Let $f$ be a separable concave piecewise linear utility function with $\kappa$ segments and $\kappa-1$ breakpoints $\hat{x}_{1}<\hat{x}_{2}<\dots<\hat{x}_{k-1}$. 
    Then the smooth concave envelope of $f$ is
    \begin{align*}
        f_{\gamma}(x) = \begin{cases}
            u_1 x + b_1 + \frac{\gamma}{2}u_1^2 & \text{if } x\in (-\infty, \hat{x}_1-\gamma u_1]\\
            u_1 \hat{x}_1 + b_1 - \frac{1}{2\gamma}(x-\hat{x}_1)^2 & \text{if } x\in [\hat{x}_1-\gamma u_1,  \hat{x}_1-\gamma u_2]\\
            u_2 x + b_2 + \frac{\gamma}{2}u_2^2 & \text{if } x\in [\hat{x}_1-\gamma u_2, \hat{x}_2-\gamma u_2]\\
            u_2 \hat{x}_2 + b_2 - \frac{1}{2\gamma}(x-\hat{x}_2)^2 & \text{if } x\in [\hat{x}_2-\gamma u_2,  \hat{x}_2-\gamma u_3]\\
            \dots & \dots\\
            u_{\kappa-1} \hat{x}_{\kappa-1} + b_{\kappa-1} - \frac{1}{2\gamma}(x-\hat{x}_{\kappa-1})^2 & \text{if } x\in [\hat{x}_{\kappa-1}-\gamma u_{\kappa-1}, \hat{x}_{\kappa-1}-\gamma u_{\kappa}]\\
            u_{\kappa} x + b_{\kappa} + \frac{\gamma}{2}u_{\kappa}^2 & \text{if } x\in [\hat{x}_{\kappa-1}-\gamma u_{\kappa}, \infty)\\
        \end{cases}
    \end{align*}
\end{repeattheorem}
\prf{We prove the result for $\kappa=2$ pieces; the result for $\kappa>2$ follows similarly. Moreau envelope of $-f(x)$ is
\begin{align}
    e_{\gamma;-f}(x)=&\min_{y}\left\{\frac{1}{2\gamma}(x-y)^2-f(y)\right\}\nonumber\\ =& \min\left\{ \min_{y\le \hat{x}_1}\left\{\frac{1}{2\gamma}(x-y)^2-u_1y-b_1\right\}, \min_{y\ge \hat{x}_1}\left\{\frac{1}{2\gamma}(x-y)^2-u_2y-b_2\right\}\right\} \label{app:separable_env_1}
\end{align}
Note that the first expression in \eqref{app:separable_env_1} attains its minimum either at its unconstrained global minimum (i.e., $x+\gamma u_1$) or at the boundary (i.e., $\hat{x}_1$). The former is attained when $x+\gamma u_1\le \hat{x}_1$ and latter is attained when $x+\gamma u_1\ge \hat{x}_1$. Similarly, the second expression in \eqref{app:separable_env_1} attains its minimum either at its unconstrained global minimum (i.e., $x+\gamma u_2$) when $x+\gamma u_2\ge \hat{x}_1$ or at the boundary (i.e., $\hat{x}_1$) when $x+\gamma u_2\le \hat{x}_1$. Given that $u_1 > u_2$, and by definition of the global minima for each segment, we can characterize the global optimal value of $y$ as follows:
\begin{align*}
    y^* = \begin{cases}
        \gamma u_1 + x & \quad\text{if } x\in (-\infty, \hat{x}_1-\gamma u_1]\\
        \hat{x}_1 & \quad\text{if } x\in [\hat{x}_1-\gamma u_1,  \hat{x}_1-\gamma u_2]\\
        \gamma u_2 + x & \quad\text{if } x\in [\hat{x}_1-\gamma u_2, \infty)\\
    \end{cases}
\end{align*}
Noting that $u_1\hat{x}_1+b_1=u_2\hat{x}_1+b_2$, we may rewrite the Moreau envelope as 
\begin{align}
    e_{\gamma;-f}(x)=\begin{cases}
        -u_1 x - b_1 - \frac{\gamma}{2}u_1^2 & \quad\text{if } x\in (-\infty, \hat{x}_1-\gamma u_1]\\
        -u_1 \hat{x}_1 - b_1 + \frac{1}{2\gamma}(x-\hat{x}_1)^2 & \quad\text{if } x\in [\hat{x}_1-\gamma u_1,  \hat{x}_1-\gamma u_2]\\
        -u_2 x - b_2 - \frac{\gamma}{2}u_2^2 & \quad\text{if } x\in [\hat{x}_1-\gamma u_2, \infty)\\
    \end{cases} \label{app:separable_env_2}
\end{align}
Multiplying by -1 gives the concave envelope of $f$.
}

\subsection{Proof of Proposition \ref{prop:initial_bound}}

\begin{repeattheorem}[Proposition \ref{prop:initial_bound}.]
    Provided the matching market model admits an integral optimal solution,  then (i) the lowerbound provided by \eqref{op:initial} is tight, and (ii) both FW and CP algorithms terminate after one iteration when initialized by the solution produced by \eqref{op:initial}.
\end{repeattheorem}
\prf{
	We prove the result for one-sided; the proof for \textit{2LF} follows similarly. 
	Let $\xx^*\in \calX$ be the optimal solution. Since $x^*_{ij}\in\{0,1\}$ and $\sum_{j\in G}x^*_{ij}=1$ for each $i$, we have $\xx^*_i=\ee_j$ for $j$ such that $x^*_{ij}=1$, implying $u_i(\xx)=\sum_{j\in G}x^*_{ij}u_i(\ee_i)$. Furthermore, optimality of $\xx^*$ implies that $x^*_{ij}=0$ for $i$ and $j$ such that $u_{ij}-c_i\le 0$. Therefore, $\log(u_i(\xx^*)-c_i)=\sum_{j\in G}x^*_{ij}\log(u_i(\ee_j)-c_i)=\sum_{j\in G}x^*_{ij}\log_M(u_i(\ee_j)-c_i)$, which also proves optimality of $\xx^*$ for \eqref{op:initial}. 

 Strict concavity of the objective functions in the market models imply that CP and FW need only one iteration to prove optimality of the initial solution when the initial solution is indeed optimal.
}

\clearpage
\section{Entropic Multidimensional Matching}\label{app:entropic_multimatching}

Algorithm~\ref{pseudo-code-sinkhorn} presents our implementation of Sinkhorn-Knopp scaling algorithm extended to multiple marginals.

\begin{algorithm}[h!]
    \textbf{Input:} Gradient tensor $\gb$, entropy weight $\zeta$
    
	\textbf{Initialization:} Set $m'=m+1$, $\epsilon'=\frac{m'\zeta\log(n)}{8\|\gb\|_{\infty}}$, and $\tilde{\gb}=\exp\left(\frac{\gb}{\zeta}\right)$, with element-wise $\exp$.

    \vspace{0.2cm}
    
    \textbf{Step 1 (Scaling):} Set $s=0$, $E^{(0)}=\infty$, and let $\rr$ and $\bbeta^{(0)}$ be  matrices of all ones in $\R^{m'\times n}$.
    
	\While{$E^{(s)}>\epsilon'$}{
        Let $\hat{\xx}^{(s)}$ be a tensor with $\hat{x}^{(s)}_{i_1,i_2,\dots,i_{m'}}=\tilde{g}_{i_1,i_2,\dots,i_{m'}} \prod_{k=1}^{m'}\beta^{(s)}_{k,i_k}$.

        Let $\bar{\rr}_k$ be the $k^{\text{th}}$ margin of $\hat{\xx}^{(s)}$, and choose marginal $$k'=\argmax_{k\in\range{m'}}\left\{\sum_{i=1}^n r_{ki}-\bar{r}_{ki}+r_{ki}\log\left(\frac{r_{ki}}{\bar{r}_{ki}}\right)\right\}.$$
  
        Compute $E^{(s)}=\sum_{k=1}^{m'}\sum_{i=1}^n|r_{ki}-\bar{r}_{ki}|$.
        
		Let $\bbeta^{(s+1)} = \bbeta^{(s)}$, and set $\beta^{(s+1)}_{k',i} = \beta^{(s)}_{k',i}\frac{r_{k',i}}{\bar{r}_{k',i}}$ for each $i\in \range{n}$.

		Increment $s\leftarrow s+1$
	}

    \vspace{0.4cm}
    
    \textbf{Step 2 (Rounding):} Let $\hat{\xx}$ be the tensor produced at the end of \textbf{Step 1}.

    \For{$k=1$ to $m'$}{
       Let $\bar{\rr}_k$ be the $k^{\text{th}}$ margin of $\hat{\xx}$ and compute $\zz\in \R^n$ with $z_i=\min\left\{1, \frac{r_{ki}}{\bar{r}_{ki}}\right\}$ for each $i\in\range{n}$.

       \For{$j=1$ to $n$}{
       Update $\hat{x}_{i_1,i_2,\dots,i_{m'}}\leftarrow z_{j}\hat{x}_{i_1,i_2,\dots,i_{m'}}$ in which $i_k=j$ is fixed.
       }
    }

    Let $\bar{\rr}_k$ be the $k^{\text{th}}$ margin of $\hat{\xx}$ and compute $\ee_{k}=\rr_{k}-\bar{\rr}_{k}$ for each $k\in\range{m'}$. 

    Update $\hat{x}_{i_1,i_2,\dots,i_{m'}}\leftarrow \hat{x}_{i_1,i_2,\dots,i_{m'}} + \frac{\prod_{k=1}^{m'} e_{k,i_k}}{\|\ee_1\|_1^{m'-1}}$

    \textbf{Output:} $\hat{\xx}$
	\caption{Entropic $(m+1)$-dimensional matching}
	\label{pseudo-code-sinkhorn}
\end{algorithm}

\section{Supplementary Numerical Results}\label{app:results}

Detailed results for {\em 1LAD}, {\em 1LF} and {\em 2LF} are given in Tables~\ref{tab:results-1lad}, \ref{tab:results-1lf} and \ref{tab:results-2lf}, respectively. 
Each entry is the average value over 30 instances (10 replication for each $\rho$), while ``Mem.'' indicates running out of memory.
For each number of agents/goods ($n$), we present the time to produce the initial solution ($t_0$) as described in Section~\ref{sec:initial_matching}, optimality gap of the initial solution ($g_0$) for FW and CP according to equations \eqref{eq:gap-fw} and \eqref{eq:gap-cp}, respectively, and total computation times ($t_{\text{T}}$). 

Figure~\ref{fig:example} illustrates an example where FW produces a rational optimal solution while CP does not.

\begin{table}[t]
\footnotesize
\centering
\noindent\begin{tabular*}{\columnwidth}{@{\extracolsep{\stretch{1}}}*{12}{@{}rrrrrrrrrrrr@{}}}
\toprule
 & \multicolumn{5}{c}{\textbf{binary}} &  & \multicolumn{5}{c}{\textbf{nonbinary}} \\ \midrule
n & \multicolumn{1}{c}{$t_0$} & \multicolumn{1}{c}{$g_0$[FW]} & \multicolumn{1}{c}{$g_0$[CP]} & \multicolumn{1}{c}{$t_{\text{T}}$[FW]} & \multicolumn{1}{c}{$t_{\text{T}}$[CP]} &  & \multicolumn{1}{c}{$t_0$} & \multicolumn{1}{c}{$g_0$[FW]} & \multicolumn{1}{c}{$g_0$[CP]} & \multicolumn{1}{c}{$t_{\text{T}}$[FW]} & \multicolumn{1}{c}{$t_{\text{T}}$[CP]} \\ \cmidrule(r){1-6} \cmidrule(l){8-12} 
10 & 0.00 & 12.50\% & 0.04\% & 0.09 & 0.00 &  & 0.00 & 15.14\% & 6.71\% & 0.01 & 0.00 \\
20 & 0.00 & \textless{}0.01\% & \textless{}0.01\% & 0.00 & 0.00 &  & 0.00 & 0.59\% & 1.31\% & 0.00 & 0.00 \\
50 & 0.00 & \textless{}0.01\% & \textless{}0.01\% & 0.00 & 0.01 &  & 0.00 & 0.05\% & 0.11\% & 0.00 & 0.01 \\
100 & 0.00 & \textless{}0.01\% & \textless{}0.01\% & 0.00 & 0.04 &  & 0.01 & 8.61\% & 0.40\% & 0.78 & 0.06 \\
200 & 0.01 & \textless{}0.01\% & \textless{}0.01\% & 0.02 & 0.07 &  & 0.00 & 0.18\% & 0.16\% & 0.04 & 0.16 \\
500 & 0.04 & \textless{}0.01\% & \textless{}0.01\% & 0.09 & 0.75 &  & 0.02 & \textless{}0.01\% & \textless{}0.01\% & 0.09 & 0.67 \\
1000 & 0.17 & \textless{}0.01\% & \textless{}0.01\% & 0.30 & 4.74 &  & 0.11 & \textless{}0.01\% & \textless{}0.01\% & 0.26 & 1.96 \\
2000 & 0.61 & \textless{}0.01\% & \textless{}0.01\% & 1.22 & 46.03 &  & 0.42 & \textless{}0.01\% & \textless{}0.01\% & 1.12 & 15.73 \\
5000 & 3.84 & \textless{}0.01\% & Mem. & 7.25 & Mem. &  & 2.45 & \textless{}0.01\% & Mem. & 6.58 & Mem. \\
10000 & 16.90 & \textless{}0.01\% & Mem. & 32.98 & Mem. &  & 8.99 & \textless{}0.01\% & Mem. & 25.90 & Mem. \\
20000 & 66.26 & \textless{}0.01\% & Mem. & 123.27 & Mem. &  & 37.18 & \textless{}0.01\% & Mem. & 107.77 & Mem. \\ \bottomrule
\end{tabular*}
\caption{Computational results for {\em 1LAD}}
	\label{tab:results-1lad}

\vspace{0.5cm}

\footnotesize
\centering
\noindent\begin{tabular*}{\columnwidth}{@{\extracolsep{\stretch{1}}}*{12}{@{}rrrrrrrrrrrr@{}}}
\toprule
 & \multicolumn{5}{c}{\textbf{binary}} &  & \multicolumn{5}{c}{\textbf{nonbinary}} \\ \midrule
n & \multicolumn{1}{c}{$t_0$} & \multicolumn{1}{c}{$g_0$[FW]} & \multicolumn{1}{c}{$g_0$[CP]} & \multicolumn{1}{c}{$t_{\text{T}}$[FW]} & \multicolumn{1}{c}{$t_{\text{T}}$[CP]} &  & \multicolumn{1}{c}{$t_0$} & \multicolumn{1}{c}{$g_0$[FW]} & \multicolumn{1}{c}{$g_0$[CP]} & \multicolumn{1}{c}{$t_{\text{T}}$[FW]} & \multicolumn{1}{c}{$t_{\text{T}}$[CP]} \\ \cmidrule(r){1-6} \cmidrule(l){8-12} 
10 & 0.00 & \textless{}0.01\% & 5.74\% & 0.00 & 0.02 &  & 0.00 & 3.69\% & 7.63\% & 0.00 & 0.00 \\
20 & 0.00 & \textless{}0.01\% & \textless{}0.01\% & 0.00 & 0.00 &  & 0.00 & 0.57\% & 1.31\% & 0.00 & 0.00 \\
50 & 0.00 & \textless{}0.01\% & \textless{}0.01\% & 0.00 & 0.01 &  & 0.00 & 0.05\% & 0.10\% & 0.01 & 0.01 \\
100 & 0.00 & \textless{}0.01\% & 0.32\% & 0.00 & 0.05 &  & 0.00 & 0.19\% & 0.18\% & 0.04 & 0.05 \\
200 & 0.00 & \textless{}0.01\% & \textless{}0.01\% & 0.01 & 0.07 &  & 0.00 & 0.14\% & 0.14\% & 0.03 & 0.14 \\
500 & 0.01 & \textless{}0.01\% & \textless{}0.01\% & 0.05 & 0.62 &  & 0.02 & \textless{}0.01\% & \textless{}0.01\% & 0.08 & 0.72 \\
1000 & 0.06 & \textless{}0.01\% & \textless{}0.01\% & 0.20 & 4.56 &  & 0.09 & \textless{}0.01\% & \textless{}0.01\% & 0.24 & 1.96 \\
2000 & 0.23 & \textless{}0.01\% & \textless{}0.01\% & 0.84 & 46.87 &  & 0.42 & \textless{}0.01\% & \textless{}0.01\% & 1.09 & 15.89 \\
5000 & 1.33 & \textless{}0.01\% & Mem. & 5.01 & Mem. &  & 2.35 & \textless{}0.01\% & Mem. & 6.75 & Mem. \\
10000 & 4.18 & \textless{}0.01\% & Mem. & 20.47 & Mem. &  & 8.43 & \textless{}0.01\% & Mem. & 26.86 & Mem. \\
20000 & 22.17 & \textless{}0.01\% & Mem. & 87.39 & Mem. &  & 36.69 & \textless{}0.01\% & Mem. & 115.67 & Mem. \\ \bottomrule
\end{tabular*}
\caption{Computational results for {\em 1LF}}
	\label{tab:results-1lf}

\vspace{0.5cm}

\footnotesize
\centering
\noindent\begin{tabular*}{\columnwidth}{@{\extracolsep{\stretch{1}}}*{12}{@{}rrrrrrrrrrrr@{}}}
\toprule
 & \multicolumn{5}{c}{\textbf{binary}} &  & \multicolumn{5}{c}{\textbf{nonbinary}} \\ \midrule
n & \multicolumn{1}{c}{$t_0$} & \multicolumn{1}{c}{$g_0$[FW]} & \multicolumn{1}{c}{$g_0$[CP]} & \multicolumn{1}{c}{$t_{\text{T}}$[FW]} & \multicolumn{1}{c}{$t_{\text{T}}$[CP]} &  & \multicolumn{1}{c}{$t_0$} & \multicolumn{1}{c}{$g_0$[FW]} & \multicolumn{1}{c}{$g_0$[CP]} & \multicolumn{1}{c}{$t_{\text{T}}$[FW]} & \multicolumn{1}{c}{$t_{\text{T}}$[CP]} \\ \cmidrule(r){1-6} \cmidrule(l){8-12} 
10 & 0.00 & \textless{}0.01\% & 1.44\% & 0.00 & 0.00 &  & 0.00 & 25.47\% & 16.06\% & 0.01 & 0.00 \\
20 & 0.00 & \textless{}0.01\% & 1.58\% & 0.00 & 0.01 &  & 0.00 & 3.84\% & 5.13\% & 0.04 & 0.01 \\
50 & 0.00 & \textless{}0.01\% & 0.17\% & 0.00 & 0.01 &  & 0.00 & 7.24\% & 3.96\% & 0.43 & 0.08 \\
100 & 0.00 & \textless{}0.01\% & 0.05\% & 0.00 & 0.27 &  & 0.00 & 1.16\% & 0.95\% & 0.46 & 0.27 \\
200 & 0.00 & \textless{}0.01\% & 0.09\% & 0.02 & 0.80 &  & 0.00 & 0.34\% & 0.25\% & 0.19 & 0.46 \\
500 & 0.01 & \textless{}0.01\% & 0.06\% & 0.11 & 3.33 &  & 0.02 & 0.08\% & 0.04\% & 0.45 & 1.74 \\
1000 & 0.06 & \textless{}0.01\% & 0.03\% & 0.38 & 35.34 &  & 0.09 & 0.03\% & 0.01\% & 0.80 & 4.69 \\
2000 & 0.25 & \textless{}0.01\% & 0.02\% & 1.60 & 1144.64 &  & 0.37 & 0.01\% & \textless{}0.01\% & 2.20 & 19.84 \\
5000 & 1.44 & \textless{}0.01\% & Mem. & 8.96 & Mem. &  & 1.96 & 0.99\% & Mem. & 128.06 & Mem. \\
10000 & 4.33 & \textless{}0.01\% & Mem. & 37.99 & Mem. &  & 7.84 & 0.20\% & Mem. & 162.64 & Mem. \\
20000 & 18.77 & \textless{}0.01\% & Mem. & 229.31 & Mem. &  & 37.37 & \textless{}0.01\% & Mem. & 239.64 & Mem. \\ \bottomrule
\end{tabular*}
\caption{Computational results for {\em 2LF}}
	\label{tab:results-2lf}
\end{table}

\begin{figure}[h]
	\centering
	\subfloat[][\centering Utility matrix]{\footnotesize
		\scalebox{0.7}{
			\setlength{\tabcolsep}{3pt}
			\begin{tabular}{|llllllllll|}
				\hline
				1 & 0 & 0 & 1 & 0 & 0 & 1 & 1 & 0 & 0 \\
				1 & 1 & 1 & 1 & 0 & 0 & 0 & 0 & 0 & 0 \\
				1 & 0 & 1 & 0 & 0 & 1 & 0 & 0 & 0 & 0 \\
				1 & 1 & 0 & 0 & 0 & 0 & 0 & 0 & 0 & 0 \\
				1 & 0 & 1 & 0 & 0 & 0 & 0 & 1 & 0 & 0 \\
				0 & 0 & 0 & 0 & 0 & 0 & 0 & 1 & 0 & 0 \\
				1 & 0 & 0 & 0 & 0 & 0 & 0 & 0 & 0 & 0 \\
				0 & 0 & 0 & 1 & 0 & 0 & 0 & 0 & 1 & 0 \\
				0 & 0 & 0 & 0 & 0 & 0 & 0 & 0 & 1 & 1 \\
				0 & 1 & 0 & 1 & 0 & 0 & 0 & 0 & 0 & 0\\
				\hline
	\end{tabular}
}}%
	\;
	\subfloat[][\centering Solution produced by FW]{\footnotesize
			\scalebox{0.7}{
				\setlength{\tabcolsep}{3pt}
				\begin{tabular}{|llllllllll|}
					\hline
					0          & 0          & 0          & 0          & 0          & 0 & 1 & 0          & 0 & 0 \\
					0.1$\bar{6}$ & 0.1$\bar{6}$ & 0.$\bar{3}$ & 0.1$\bar{6}$ & 0.1$\bar{6}$  & 0 & 0 & 0          & 0 & 0 \\
					0          & 0          & 0          & 0          & 0          & 1 & 0 & 0          & 0 & 0 \\
					0          & 0.8$\bar{3}$ & 0          & 0          & 0.1$\bar{6}$  & 0 & 0 & 0          & 0 & 0 \\
					0          & 0          & 0.$\bar{6}$ & 0          & 0.1$\bar{6}$ & 0 & 0 & 0.1$\bar{6}$ & 0 & 0 \\
					0          & 0          & 0          & 0          & 0.1$\bar{6}$ & 0 & 0 & 0.8$\bar{3}$ & 0 & 0 \\
					0.8$\bar{3}$ & 0          & 0          & 0          & 0.1$\bar{6}$ & 0 & 0 & 0          & 0 & 0 \\
					0          & 0          & 0          & 0          & 0          & 0 & 0 & 0          & 1 & 0 \\
					0          & 0          & 0          & 0          & 0          & 0 & 0 & 0          & 0 & 1 \\
					0          & 0          & 0          & 0.8$\bar{3}$ & 0.1$\bar{6}$ & 0 & 0 & 0          & 0 & 0\\
					\hline
	\end{tabular}}}%
	\;
	\subfloat[][\centering Solution produced by CP]{\footnotesize
			\scalebox{0.7}{
				\setlength{\tabcolsep}{2pt}
				\begin{tabular}{|llllllllll|}
					\hline
					0          & 0          & 0          & 0          & 0          & 0 & 1 & 0          & 0 & 0 \\
					0.16666333 & 0.16658701 & 0.33291273 & 0.16694404 & 0.1668929  & 0 & 0 & 0          & 0 & 0 \\
					0          & 0          & 0          & 0          & 0          & 1 & 0 & 0          & 0 & 0 \\
					0          & 0.83332482 & 0          & 0          & 0.1666716  & 0 & 0 & 0          & 0 & 0 \\
					0          & 0          & 0.66708727 & 0          & 0.16645611 & 0 & 0 & 0.16645662 & 0 & 0 \\
					0          & 0          & 0          & 0          & 0.16645662 & 0 & 0 & 0.83354338 & 0 & 0 \\
					0.83333309 & 0          & 0          & 0          & 0.16666691 & 0 & 0 & 0          & 0 & 0 \\
					0          & 0          & 0          & 0          & 0          & 0 & 0 & 0          & 1 & 0 \\
					0          & 0          & 0          & 0          & 0          & 0 & 0 & 0          & 0 & 1 \\
					0          & 0          & 0          & 0.83305596 & 0.16685587 & 0 & 0 & 0          & 0 & 0\\
					\hline
	\end{tabular}}}%
	\caption{An example where FW produces a rational optimal solution for an instance of {\em 1LF}.}%
	\label{fig:example}
\end{figure}

\end{APPENDICES}

\end{document}